\documentclass[11pt]{article}
\usepackage[a4paper,portrait,lmargin=2.5cm,rmargin=2.5cm,tmargin=2.5cm,bmargin=2.5cm]{geometry}
\usepackage[utf8]{inputenc}
\usepackage[russian,english]{babel}

\usepackage{mathtools}
\usepackage{amsmath,amssymb,amsthm}
\usepackage{stmaryrd} 
\usepackage{bbm} 
\usepackage{bbold} 
\usepackage{enumitem}
\theoremstyle{plain}
\newtheorem{theorem}{Theorem}
\newtheorem{lemma}{Lemma}
\newtheorem{proposition}[lemma]{Proposition}

\newtheorem{corollary}[lemma]{Corollary}
\theoremstyle{definition}
\newtheorem{definition}{Definition}
\numberwithin{equation}{section}
\setlist[enumerate,1]{label=\textnormal{(\roman*)}}

\usepackage[dvipsnames]{xcolor}
\definecolor{darkgreen}{rgb}{0,.3,0}
\definecolor{darkblue}{rgb}{0,0,.5}
\definecolor{darkred}{rgb}{.4,0,0}

\newcommand\scalemath[2]{\scalebox{#1}{\mbox{\ensuremath{\displaystyle #2}}}}

\usepackage{hyperref}

\hypersetup{
pdfauthor={Paul-Hermann Balduf, Simone Hu},
pdfcreator={LaTeX with hyperref},
colorlinks=true,
linkcolor=darkred,
citecolor=darkgreen,
filecolor=black,
urlcolor=darkblue,
bookmarks=true,
plainpages=false,
pdfpagelabels=true,
breaklinks=true
}

\usepackage[nameinlink]{cleveref}
\crefname{figure}{figure}{figures}

\usepackage[font=small,labelfont=bf]{caption}
\usepackage[font=footnotesize]{subcaption}
\usepackage{placeins} 
\usepackage{csquotes} 
\usepackage{orcidlink}

\usepackage[
	backend=biber,
	style=numeric,
	date=year,
	eprint=true,
	isbn=false,
  doi=false,
  giveninits=true,
  maxnames=9
]{biblatex}
\usepackage[hyphens]{xurl}

\AtEveryBibitem{
	\clearfield{urlyear}
	\clearfield{urlmonth}
	\clearfield{pubstate}
}

\renewbibmacro{in:}{}
\newbibmacro{string+doiurlisbn}[1]{%
  \iffieldundef{doi}{%
    \iffieldundef{url}{%
     \iffieldundef{eprint}{%
      \iffieldundef{isbn}{%
        \iffieldundef{issn}{%
            #1%
        }{%
          \href{https://books.google.com/books?vid=ISSN\thefield{issn}}{#1}%
      	}
      }{%
        \href{https://books.google.com/books?vid=ISBN\thefield{isbn}}{#1}%
      }%
     }{%
     \href{https://arxiv.org/abs/\thefield{eprint}}{#1}%
     }
    }{%
      \href{\thefield{url}}{#1}%
    }%
  }{%
    \href{https://dx.doi.org/\thefield{doi}}{#1}%
  }%
}
\DeclareFieldFormat[article,incollection,inproceedings,webpage]{title}{\usebibmacro{string+doiurlisbn}{\mkbibquote{#1}}}
\DeclareFieldFormat[article,incollection,inproceedings]{url}{ }
\DeclareFieldFormat[unpublished]{title}{ \mkbibquote{#1}}
\DeclareFieldFormat{postnote}{#1}
\DeclareFieldFormat{eprint}{\href{https://www.arxiv.org/abs/#1}{\mkbibacro{arXiv}:#1}}

\newbibmacro{string+doi}[1]{
    \iffieldundef{doi}{
    	\iffieldundef{url}{
    		#1
    	}{
        		\href{\thefield{url}}{#1}
        }
    } {
        	\href{https://dx.doi.org/\thefield{doi}}{#1}
    }
}

\DeclareFieldFormat{url}{ 
	\iffieldundef{eprint:arxiv}{
		\iffieldundef{eprint}{
			\iffieldundef{doi}{ 
				\mkbibacro{URL}\addcolon\space{ \url{#1}}
			}{} 
		}{ 
		}
	}{
		\mkbibacro{arXiv}: \addcolon\space\url{\thefield{eprint:arxiv}}
	}
}

\usepackage{graphicx}
\usepackage{tikz}
\usetikzlibrary{arrows.meta}
\usetikzlibrary{calc}
\usetikzlibrary{decorations.markings}

\tikzstyle{edge} = [black,line width=.25mm]
\tikzstyle{vertex}=[circle,minimum size=1.5mm, draw=black, fill=black, inner sep=0mm]

\usepackage[framemethod=tikz,
roundcorner=2pt,
linewidth=.5pt,
skipabove=10pt plus 20pt minus 2pt,
skipbelow=12pt plus 20pt minus 2pt]{mdframed}
\mdfsetup{linewidth=1.0pt, roundcorner=5pt}

\newcounter{example}

\NewDocumentEnvironment {example} {o}
{ \refstepcounter{example}
	\pagebreak[2]
  \begin{mdframed}[linewidth=0.8pt,  linecolor=black,
		bottomline=false,topline=false,rightline=false, startcode=\needspace{4\baselineskip}]
    \noindent \textbf{Example  \theexample.}~}
	{ \end{mdframed}
}

\newcommand{\Graph}{G}
\newcommand{\laplacian}{\mathbbm{L}}

\newcommand{\duallaplacian}{\Lambda}
\newcommand{\cyclelaplacian}{\duallaplacian}
\newcommand{\explaplacian}{\mathbbm{M}}
\newcommand{\identitymatrix}{\mathbb{1}}
\newcommand{\zeromatrix}{\mathbb{0}}
\newcommand{\incidencematrix}{\mathbbm{I}}

\newcommand{\edgematrix}{\mathbbm{D}}
\newcommand{\cycleincidencematrix}{\mathcal{C}}
\newcommand{\cyclebasis}{\cycleincidencematrix}
\newcommand{\fundamentalcbasis}{\mathcal{T}}
\newcommand{\pathmatrix}{\mathcal{P}}
\newcommand{\dodgson}{\psi}
\newcommand{\firstsymanzik}{\dodgson}

\newcommand{\loopnumber}{\ell}

\newcommand{\cyclespace}{{H_1}}
\newcommand{\cutspace}{{F}}

\DeclarePairedDelimiter{\abs}{\lvert}{\rvert}
\DeclarePairedDelimiter{\inner}{\langle}{\rangle}
\DeclareMathOperator{\Pf}{Pf}
\DeclareMathOperator{\Aut}{Aut}
\DeclareMathOperator{\sgn}{sgn}
\DeclareMathOperator{\haf}{haf}
\DeclareMathOperator{\diag}{diag}

\renewcommand{\d}{ \,\mathrm{d}}
\newcommand{\defas}{\coloneqq}
\newcommand{\Transpose}{\intercal}
\newcommand{\mb}{\mathbbm}
\newcommand{\mc}{\mathcal }
\newcommand{\Z}{\mathbbm{Z}}
\newcommand{\R}{\mathbbm{R}}
\newcommand{\Q}{\mathbbm{Q}}
\newcommand{\GL}{\mathsf{GL}}
\newcommand{\concatm}[2]{\left( #1 \,\vert\, #2 \right)}
\newcommand{\cGC}{\mathsf{GC}}

\newcommand{\longequation}{ \crefname{equation}{equation}{equations}}
\newcommand{\shortequation}{  \crefname{equation}{eq.}{eqs.}}

\usetikzlibrary{tilings.penrose}
\newcommand{\penroseB}{
	\begin{tikzpicture}[
		every thick rhombus/.style={
			every rhombus/.try,
			draw=black,
			fill=white,
		},
		every tile pic/.style={scale=.08},
		]
		\path (-.02,0) -- (.05,0);
		\pic[rotate=-18,thick rhombus,name=a0];
		\foreach[evaluate=\k as \kmo using int(\k-1)] \k in {1,...,4}
		{
			\pic[thick rhombus,name=a\k,align with={a\kmo} along A];
		}
	\end{tikzpicture}
}

\newcommand{\penroseD}{
	\begin{tikzpicture}[
		every thick rhombus/.style={
			every rhombus/.try,
			draw=black,
			fill=white,
		},every thin rhombus/.style={
			every rhombus/.try,
			draw=black,
			fill=white,
		},
		every tile pic/.style={scale=.07},
		]
		\path (-.1,0) -- (.1,0);

		\pic[rotate= 126,thick rhombus,name=t0];

		\pic[ thin rhombus,name=f1,align with=t0 along a];
		\pic[ thin rhombus,name=f2,align with=t0 along A];
		\pic[ thick rhombus,name=t1,align with=f2 along B];
		\pic[ thick rhombus,name=t2,align with=t1 along A];
		\pic[ thick rhombus,name=t3,align with=t2 along A];
		\pic[ thin rhombus,name=f3,align with=t1 along a];
		\pic[ thin rhombus,name=f4,align with=t3 along B];
		\pic[ thin rhombus,name=f5,align with=f4 along a];
		\pic[ thick rhombus,name=t4,align with=t0 along b];
	\end{tikzpicture}
}

\makeatletter
\def\@fnsymbol#1{\ensuremath{\ifcase#1   \or%
		\protect\penroseD %
		\or%
		\protect\penroseB %
		\or \dagger 		 \or \ddagger\or
		\mathsection \else\@ctrerr\fi}}
\makeatother

\title{The Topological  form is the   Pfaffian form}
\newcommand{\email}[1]{\href{mailto:#1}{#1}}
\author{
  Paul-Hermann Balduf%
  \,\thanks{Mathematical Institute, University of Oxford, OX2 6GG, UK, \orcidlink{0000-0003-4475-3031}  \email{paul-hermann.balduf@maths.ox.ac.uk}}~
  ~and
  Simone Hu%
  \thanks{Mathematical Institute, University of Oxford, OX2 6GG, UK, \orcidlink{0009-0002-0063-3695} \email{simone.hu@maths.ox.ac.uk}}
}

\begin{document}

\maketitle

\renewcommand{\thefootnote}{\arabic{footnote}}

\begin{abstract}
	\noindent
	For a given graph $\Graph$, 	Budzik, Gaiotto, Kulp, Wang, Williams, Wu,  Yu, and the first author studied a \enquote{topological} differential form $\alpha_\Graph$, which expresses violations of BRST-closedness of a quantum field theory along a single topological direction.
	In a seemingly unrelated context,   Brown, Panzer, and the second author studied a \enquote{Pfaffian} differential form $\phi_\Graph$, which is used to construct cohomology classes of the odd commutative graph complex.
	We  give an explicit combinatorial proof that  $\alpha_\Graph$ coincides with  $\phi_\Graph$.   We also discuss the equivalence of several  properties of these forms, which had been established independently for both contexts in previous work.

\end{abstract}

\tableofcontents

\section{Introduction and results}

In this paper, a \enquote{graph $\Graph$} will always mean  a connected, finite, not necessarily simple (self-loops and multiedges are allowed) graph, with directed edges $E_\Graph = \left \lbrace e_1, \ldots, e_{\abs{E_\Graph}} \right \rbrace $ and  vertices $V_\Graph = \left \lbrace v_1, \ldots, v_{\abs{V_\Graph}} \right \rbrace $. The  labellings and the directions of edges are arbitrary, but fixed. \\
The notation $A \defas B$ means that $A$ is defined to be equal to $B$.

\subsection{The main result}
The authors of \cite{gaiotto_higher_2024} introduced a \emph{topological differential form} associated to a graph $\Graph$
\begin{align*}
	\alpha_\Graph \defas \frac{1}{\pi^{\frac{\abs{E_\Graph}}{2}}}\idotsint \limits_{\R^{\abs{V_\Graph}-1}}\bigwedge_{e\in E_\Graph}e^{-s_e^2} \;  \d s_e,
\end{align*}
which is the parametric Feynman integrand to compute quantum corrections to BRST-closedness in a 1-dimensional topological QFT.
We will give more details and background of the topological form in \cref{sec:intro_alpha}.
\begin{example}{}\label{ex:dunces_cap}
  Let $G$ be the dunce's cap graph, shown on the right.  \\[-.12cm]
	\begin{minipage}{.65\linewidth}
		Its  edges and vertices are labelled as indicated.
		Each of the four edges corresponds to one Schwinger parameter $a_e$. The topological form of this graph has been computed in \cite[Example 13]{balduf_combinatorial_2024}.
	\end{minipage}\hfill
	\begin{minipage}{.25\linewidth}
		\vspace{-.3cm}
		 \begin{tikzpicture}[baseline=-3]
		 	\node[vertex, label=below:$v_1$](v1) at (1,1){};
		 	\node[vertex, label=below:$v_2$ ](v2) at (-.5,-.5){};
		 	\node[vertex, label=below:$v_3$ ](v3) at (2.5,-.5){};

		 	\draw[edge, >=Triangle, decoration={markings, mark=at position .85 with {\arrow{>}}}, postaction={decorate}, ](v2) to node[pos=.5, left]{$a_1$}   (v1) ;
		 	\draw[edge, >=Triangle, decoration={markings, mark=at position .25 with {\arrow{<}}}, postaction={decorate}, ](v3) to node[pos=.5, right]{$a_2$}(v1);
		 	\draw[edge, >=Triangle, decoration={markings, mark=at position .25 with {\arrow{<}}}, postaction={decorate},  bend angle=30,bend left](v3) to node[pos=.5, below]{$a_4$}(v2);
		 	\draw[edge, >=Triangle, decoration={markings, mark=at position .25 with {\arrow{<}}}, postaction={decorate},   bend angle=20, bend right](v3) to node[pos=.5, below]{$a_3$}(v2);
		 \end{tikzpicture}
	\end{minipage}

	\begin{align*}
    \scalemath{0.95}{\alpha_G = \frac{ -a_4 (\d a_1 \wedge \d a_3 + \d a_2 \wedge \d a_3) + a_3 (\d a_1 \wedge \d a_4 + \d a_2 \wedge \d a_4) -  (a_1 + a_2) \d a_3 \wedge \d a_4}{8\pi(a_1 a_3 + a_2 a_3 + a_1 a_4 + a_2 a_4 + a_3 a_4)^{3/2}}}.
	\end{align*}
\end{example}
On the other hand, the authors of \cite{brown_unstable_2024} introduced a \emph{Pfaffian differential form}
	\begin{align*}
	\phi_\Graph  &\defas \frac{1}{(-2\pi)^{\frac \loopnumber 2}} \frac{\Pf\left( \d \duallaplacian_{\cycleincidencematrix} \cdot \duallaplacian_{\cycleincidencematrix}^{-1} \cdot \d \duallaplacian_{\cycleincidencematrix} \right)   }{\sqrt{\det{\duallaplacian_{\cycleincidencematrix}}}},
\end{align*}
whose integrals contain information about the cohomology of the odd graph complex $\cGC_3$. We will review background and definitions in \cref{sec:intro_phi}.
\begin{example}\label{ex:dunces-pff}
  Using \cite[Example 2.3]{brown_unstable_2024}, the Pfaffian form for the dunce's cap graph with
  \[ \cyclelaplacian_{\cyclebasis} =
    \begin{pmatrix}
			a_1 + a_2 + a_4 & a_4 \\
			a_4 & a_3 + a_4
    \end{pmatrix}, \quad\text{turns out to be exactly}\quad \phi_G = (-4) \cdot \alpha_G. \]
\end{example}

The main result of the present work is that the topological form  and the Pfaffian form coincide.
\begin{theorem}\label{thm:main_thm}
	Let $\Graph$ be a connected graph with loop number $\loopnumber \geq 0$, and let $\phi_\Graph$ be its Pfaffian form (\cref{def:pff-form}) and $\alpha_\Graph$ its topological form (\cref{def:alpha}). Then
	\begin{align*}
		\alpha_\Graph &= \frac{\det\concatm{\cycleincidencematrix}{\pathmatrix}}{2^{\loopnumber}}\cdot \phi_\Graph.
	\end{align*}
	Here, $\concatm{\cycleincidencematrix}{\pathmatrix}  $ is the concatenation of the cycle incidence matrix $\cycleincidencematrix$ (\cref{def:cycle_incidence_matrix}) used to define $\phi_{\Graph}$ and a path matrix $\pathmatrix$ (\cref{def:pathmatrix}). The determinant ~$\det \concatm{\cycleincidencematrix}{\pathmatrix}  \in \left \lbrace +1,-1 \right \rbrace  $ depends only on choices of labellings, edge directions, and the cycle basis of $\Graph$, and is independent of the choice of $\pathmatrix$, see \cref{sec:concatenated_matrices}.
\end{theorem}

The proof of the theorem requires some technical lemmas, and a new formula for the Pfaffian form in terms of Dodgson polynomials, which will be established in \cref{sec:definitions_proofs}.
In order to make the article accessible to a broad audience, that section also includes a rather elementary introduction and review of the required concepts and definitions from graph theory.

Before coming to the proof, we review the background of the topological form $\alpha_\Graph$ in \cref{sec:intro_alpha}, and the background of the Pfaffian form $\phi_\Graph$ in \cref{sec:intro_phi}. Then, in \cref{sec:consequences}, we  comment on the consequences that follow from \cref{thm:main_thm}. In particular, we demonstrate that statements that had been proven and/or observed on either side have their natural counterparts on the other side of the equivalence given by \cref{thm:main_thm}.

\subsection[The topological form alpha]{The topological form $\alpha_\Graph$}\label{sec:intro_alpha}

In \cite{budzik_feynman_2023,gaiotto_higher_2024,balduf_combinatorial_2024,wang_factorization_2024}, a differential form $\alpha_\Graph$ is examined, which computes quantum corrections to a BRST-differential in a topological theory.

The BRST-formalism is a way to systematically handle gauge redundancies in  field theories, see e.g. \cite{henneaux_lectures_1990} for a physical or \cite{figueroa-ofarrill_brst_2006} for a mathematical exposition. The idea is to construct a sequence of \enquote{ghost fields} and \enquote{antifields}, together with a \emph{BRST differential} $Q_\text{class}$ which acts as a differential operator on fields. $Q_\text{class}$ turns the space of fields into a chain complex such that physical, gauge-invariant observables $\mc O_\text{phys}$ are represented as the 0\textsuperscript{th} cohomology group, that is, $Q_\text{class}\mc O_\text{phys}=0$ and there is no  $\tilde{\mc O}$ with the property that $\mc O_\text{phys} = Q_\text{class} \tilde{\mc O}$.

Even when an observable is BRST-closed in a classical field theory, (self-) interactions in the quantized theory can potentially lead to violations of classical BRST closedness at the quantum level. In
\cite{gaiotto_higher_2024}, a formalism is developed to compute these anomalies perturbatively. Starting from a free quantum field theory with BRST differential $Q_\text{free}$, one adds an interaction term $\mc I$ to the action. This new interaction contributes to Feynman integrals, and the result may or may not be closed under the original operator $Q_\text{free}$. If $\mc O_1, \ldots, \mc O_n$ are local operators (which amount to the vertices of Feynman diagrams), the violation of $Q$-closedness schematically defines a \emph{bracket}
\begin{align}\label{def:bracket}
\left \lbrace \mc O_1, \ldots, \mc O_n \right \rbrace &\defas Q_\text{free} \left( \int_{\mb R^{d(n-1)}} \mc O_1 \cdots \mc O_n \right)  .
\end{align}

The first bracket $\left \lbrace \mc I \right \rbrace = Q_\text{free}\mc I$ describes leading (i.e. classical) violation of BRST symmetry by $\mc I$, and one typically only considers such $\mc I$ where $\left \lbrace \mc I \right \rbrace =0$. Then,   $\left \lbrace \mc I, \cdot \right \rbrace $ describes the BRST-violation introduced by $\mc I$ at leading quantum order, $\left \lbrace \mc I, \mc I, \cdot \right \rbrace $ at next order, etc. Conversely, these are the \enquote{correction} terms by which $Q_\text{free}$ needs to be modified in order to obtain a differential $Q_\text{quantum}$ which describes an unbroken symmetry in presence of $\mc I$, schematically
\begin{align*}
	Q_\text{quantum}X := Q_\text{free}- \left \lbrace \mc I,X \right \rbrace - \left \lbrace \mc I,\mc I, X \right \rbrace - \ldots.
\end{align*}
The BRST-formalism requires that $Q_\text{quantum}^2\equiv 0$. This amounts to an infinite set of \emph{quadratic relations} between the brackets $\left \lbrace \mc I, \ldots \right \rbrace $, which thereby  obtain the structure of a $L_\infty$-algebra.  One may also view the brackets from the perspective of factorization algebras \cite{costello_factorization_2023}, which, in the topological case, relate to the little-disks operad, see e.g. \cite{beem_secondary_2020,amabel_notes_2023} or the appendix of \cite{gaiotto_higher_2024}.

Generically, operators $\mc O_j$   depend on spacetime coordinates, and this dependence has to be taken into account in a particular way when computing the brackets $\left \lbrace \mc O_j, \ldots \right \rbrace $, see \cite[Sec.~5]{budzik_feynman_2023}. For the present case of a \emph{topological} theory, the only dependence is on how many (and which)  external legs a Feynman diagram has, but not their spacetime coordinates. Let $\varphi_{i,v}$ be the external (=non Wick-contracted) fields at vertices $v$, then the bracket corresponds to a sum of Feynman integrals $I_\Graph$ with symmetry factor $\abs{\Aut(\Graph)}^{-1}$,   schematically \cite[Eq.~(4.11)]{gaiotto_higher_2024},
\begin{align}\label{bracket_integral}
	\left \lbrace \mc O_1, \mc O_2, \ldots  \right \rbrace &=\sum_{\Graph} \frac{1}{\abs{\Aut(\Graph)}}I_\Graph \prod_{v\in V_\Graph} \prod_i \varphi_{i,v}.
\end{align}
In \cite{budzik_feynman_2023}, a particular choice of integration variables for the Feynman integrals $I_\Graph$ was introduced. On the one hand, it regularizes the integrals, and on the other hand, it realizes the quadratic relations between integrals $I_\Graph$  as an explicit geometric factorization of their integration domain, called the \emph{operatope}.
These variables start from a position-space representation, where $x^+_e, x^-_e \in \mb R$ are the positions of the end vertices of an edge $e$. Note that as $\alpha_\Graph$ describes a single topological direction, these $x^\pm_e$ are not vectors. One then introduces a Schwinger parameter $a_e$ for each edge, and lets
\begin{align}\label{def:se}
	s_e &\defas \frac{x^+_e-x^-_e}{\sqrt{a_e}} \in \R.
\end{align}

\begin{definition}\label{def:alpha}
	With $s_e$ from \cref{def:se}, the \emph{topological form} is the parametric integrand of the   integrals $I_\Graph=\int \alpha_\Graph$ in \cref{bracket_integral}, given explicitly by   \cite[Eq.~(3.53)]{gaiotto_higher_2024},
	\begin{align*}
		\alpha_\Graph \defas \frac{1}{\pi^{\frac{\abs{E_\Graph}}{2}}}\idotsint \limits_{\R^{\abs{V_\Graph}-1}}\bigwedge_{e\in E_\Graph}e^{-s_e^2} \;  \d s_e.
	\end{align*}
\end{definition}

This $\alpha_\Graph$ is a differential form in Schwinger parameters $\left \lbrace a_e \right \rbrace $ of degree equal to the loop number\footnote{The loop number is the number of linearly independent cycles (\cref{def:loopnumber}); a tree is a graph without cycles.} of the graph. The power of $\pi$ is chosen such that   $\alpha_\Graph=\pm 1$ if $\Graph$ is a tree, and the   integral over   parametric space, $I_\Graph =\int \alpha_\Graph$, does not involve any further normalization factors.

Despite its simple appearance, $\alpha_{\Graph}$ is generally a complicated polynomial in $a_e$ and $\d a_e$.
A key result of \cite{balduf_combinatorial_2024} was the following explicit formula (definitions will be introduced in \cref{sec:definitions_proofs}):
\begin{proposition}[Theorem 1 of \cite{balduf_combinatorial_2024}]\label{lem:dbarT}
  Let $T$ be a spanning tree in $G$ and denote the complement by  $\{f_1,\ldots,f_{\loopnumber}\} \defas \overline{T}$ , where the edges $f_j$ come in the increasing order induced by their labelling.
  Then,
  \begin{align*}
	\alpha_\Graph &=
  \frac{1}{ \pi^{\frac \loopnumber 2} 4^{\loopnumber} \left( \frac{\loopnumber}{2} \right) ! \cdot \firstsymanzik_{\Graph}^{\frac{\loopnumber+1}{2}}}
  \sum_{\substack{T \textnormal{ spanning} \\ \textnormal{tree}}}
  \det \left( \incidencematrix[T] \right)
    \left(\sum_{\sigma \in \mathfrak{S}_{\overline T}} \dodgson_{\Graph}^{ \sigma(f_1) ,\sigma(f_2)    }\cdots \dodgson_{\Graph}^{ \sigma(f_{\loopnumber-1}), \sigma(f_\loopnumber)   } \right) \bigwedge_{f\not\in T} \d a_f.
	\end{align*}
  Here, $\incidencematrix$ is the vertex incidence matrix of $G$ (\cref{def:vertex_incidence_matrix}), $\dodgson_{\Graph}^{e_i,e_j}$ are edge-indexed Dodgson polynomials (\cref{def:Dodgson_polynomial}), and the sum goes over all permutations $\sigma$ of the $\loopnumber$ edges in $\overline T$.
\end{proposition}
\noindent
A Mathematica program implementing \cref{lem:dbarT} has been published together with \cite{balduf_combinatorial_2024} and is available from the first author's homepage\footnote{\href{https://paulbalduf.com/research}{paulbalduf.com/research}  }.

\begin{example}{}\label{ex:dunces_cap_tree}
  The dunce's cap graph has five spanning trees.
  For example, with edge directions and labels as in \cref{ex:dunces_cap},  the contribution of the spanning tree $T=\left \lbrace 2,4 \right \rbrace $ is
    \begin{align*}
      \frac{(+1)}{16\pi  (a_1 a_3 + a_2 a_3 + a_1 a_4 + a_2 a_4 + a_3 a_4)^{3/2}}  \cdot (-a_4- a_4) \d a_1 \wedge \d a_3.
		\end{align*}
 This produces the summand $- a_4 \d a_1 \wedge \d a_3$ of the numerator in \cref{ex:dunces_cap}, and a sum over all five spanning trees recovers the full $\alpha_\Graph$.
\end{example}

\subsection[The Pfaffian form phi]{The Pfaffian form $\phi_\Graph$}\label{sec:intro_phi}

The \emph{Pfaffian} of a skew-symmetric $2n\times 2n$ matrix $M$ is defined as the following homogeneous polynomial of degree $n$ in the entries of $M$:
\begin{equation} \label{eq:pfaff-defn}
	\Pf(M) = \frac{1}{2^n n!} \sum_{\sigma\in\mathfrak{S}_{2n}} \sgn \sigma \cdot M_{\sigma(1),\sigma(2)}\cdots M_{\sigma(2n-1),\sigma(2n)},
\end{equation}
where the sum is over the symmetric group $\mathfrak{S}_{2n}$ of permutations of $\{1,\ldots,2n\}$. For a square matrix of odd dimension $(2n+1)\times(2n+1)$, we set $\Pf(M)=0$. It follows that $\Pf(M)^2 = \det(M)$.

\begin{definition}\label{def:pff-form}
	Let $\cyclelaplacian_\cycleincidencematrix$ be the cycle Laplacian  (\cref{def:cycle_laplacian}) of a graph and $\d \cyclelaplacian_\cycleincidencematrix$ its differential with respect to Schwinger parameters $\left \lbrace a_e \right \rbrace _{e\in \Graph} $.
	 The \emph{Pfaffian form} (\cite[Def.~2.1,~5.4]{brown_unstable_2024}) is
	\begin{align*}
		\phi_\Graph = \phi_{\duallaplacian_{\cycleincidencematrix}} &\defas \frac{1}{(-2\pi)^{\frac \loopnumber 2}} \frac{\Pf\left( \d \duallaplacian_{\cycleincidencematrix} \cdot \duallaplacian_{\cycleincidencematrix}^{-1} \cdot \d \duallaplacian_{\cycleincidencematrix} \right)   }{\sqrt{\det{\duallaplacian_{\cycleincidencematrix}}}}.
  \end{align*}\\[-0.25cm]
\end{definition}

As with $\alpha_\Graph$, this definition includes an appropriate normalization factor such that   when integrating over Schwinger parameters, $I_{\Graph}=\int \phi_\Graph$ does not require any further normalization (cf.~\cite[Def.~6.2]{brown_unstable_2024}). When $\Graph$ is a tree, then   $\phi_\Graph=+1$. Recall that we demand $\Graph$ to be connected. For disconnected $\Graph$, the Pfaffian form factorizes, while the \cref{def:alpha} of $\alpha_\Graph$ would need to be slightly changed to factorize.
Specifically, to allow for disconnected graphs, the integration domain should be modified to $\mathbb{R}^{\abs{V_{\Graph}}-c}$, where $c$ is the number of connected components of $\Graph$.

\begin{example}{}\label{ex:pff-dipole}
	The theta graph $D_3$, or 3-edge dipole, has $\abs{V_{D_3}}=2$ vertices, $\abs{E_{D_3}}=3$ edges and $\ell=2$ loops.
	The cycle Laplacian for the cycle basis $\cyclebasis=(C_1,C_2)$ as indicated in
	\begin{equation*}
		D_3=
		\begin{tikzpicture}[baseline=-3]
			\node[vertex,label={[label distance=-1mm]180:{\small$v_1$}}](v0) at (0,0){};
			\node[vertex, label={[label distance=-1mm]360:{\small$ v_2$}}](v1) at (2,0){};
			
			\draw[edge, >=Triangle, decoration={markings, mark=at position .55 with {\arrow{>}}}, postaction={decorate},  bend angle=75,looseness=1.35,bend left](v0)     to  node[pos=.1, above]{\small $1$} (v1);
			\draw[edge,  >=Triangle, decoration={markings, mark=at position .55 with {\arrow{>}}}, postaction={decorate},   bend angle=0,bend left](v0) to node[pos=.15, label={[label distance=-2mm]90:{\small$ 3$}}  ]{ }(v1);
			\draw[edge,  >=Triangle, decoration={markings, mark=at position .55 with {\arrow{>}}}, postaction={decorate},  bend angle=75,looseness=1.35,bend right](v0) to node[pos=.1, below]{\small $2$}(v1);
			
			\draw[edge,red,>=Triangle] (1,.42) +(330:.38 and .28) arc ( 330:10:.38 and .28);
			\draw[edge,red, >=Triangle, ->] (1,.42)++(20:.36) -- ++(.06,-.1);
			\draw[edge,red,>=Triangle] (1,-.42) +(10:.38 and .28) arc ( 10:330:.38 and .28);
			\draw[edge,red, >=Triangle, ->] (1,-.42)++(330:.36) -- ++(.09,.1);
			\node[red] at (1,.40){\small $C_1$};
			\node[red] at (1,-.40){\small $C_2$};
		\end{tikzpicture}
		\quad\text{is}\quad
		\cyclelaplacian_{\cyclebasis} =
		\begin{pmatrix}
			a_1 + a_3 & a_3 \\
			a_3 & a_2 + a_3
		\end{pmatrix},
		\quad\text{where}\quad
		\cyclebasis =
		\begin{pmatrix}
			1 & 0 \\
			0 & 1 \\
			-1 & -1
		\end{pmatrix},
	\end{equation*}
	in terms of the Schwinger parameters $a_1,a_2,a_3$.
	The Pfaffian form of $D_3$ is thus the following projective 2-form:
	\begin{equation}\label{eq:pff-D3}
		\phi_{D_3} = \frac{1}{2\pi} \cdot \frac{a_1 \d a_2 \wedge \d a_3 - a_2 \d a_1 \wedge \d a_3 + a_3 \d a_1 \wedge \d a_2}{(a_1a_2 + a_2a_3 + a_1a_3)^{3/2}}
		.
	\end{equation}
\end{example}

In~\cite{brown_unstable_2024}, the Pfaffian form is studied in two different contexts.
On the one hand, the definition of the Pfaffian form can be generalized to any positive definite symmetric matrix $X$ (in place of $\cyclelaplacian_{\cyclebasis}$) and in this context, $\phi_X$ is used to construct non-zero classes in the cohomology of $\GL_{2n}(\Z)$, the group of $2n\times 2n$ matrices over $\mb Z$.

On the other hand, and more relevant to this paper, when applied to graphs and cycle Laplacian matrices, $\phi_\Graph$ is used to construct cocycles for the \emph{odd commutative graph complex}.

The commutative graph complexes $\cGC_N$  are combinatorially defined chain complexes introduced by Kontsevich in~\cite{kontsevich_formal_1993, kontsevich_feynman_1994}.
These graph complexes, and their variants, have been studied in numerous different contexts; to name a few, they have appeared in the study of Lie algebras and deformation theory~\cite{willwacher_kontsevichs_2015}, of perturbative Chern-Simons theory and knot invariants~\cite{bar-natan_vassiliev_1995,duzhin_algebra_1998} and of certain moduli spaces.
The different variants of graph complexes arise from the types of decorations allowed on a graph, for example graphs with hairs, ribbon graphs, pairs of graphs and spanning forests etc., as well as the type of \emph{orientation data} associated to each graph.
For example, one type of orientation is specifying an ordering of the edges of the graph.
The commutative graph complexes are the simplest versions, where the graphs come with no extra decorations.

Up to isomorphism and degree shifts, there are only two commutative graph complexes denoted by $\cGC_N$, differing only in the type of orientation data being considered.
Specifically, $\cGC_N$ is a $\Q$-vector space spanned by isomorphism classes of \emph{oriented graphs}, graphs equipped with orientation data.
The subscript $N$ determines the degree grading of the graphs, and its parity determines which orientation data is being used.
This vector space is turned into a chain complex by the differential $\partial \Graph = \sum_{e} \Graph / e$, which sums over all ways to contract an edge in an oriented graph.
It is bigraded by \emph{degree} $k = \abs{E_{\Graph}} - N\loopnumber$ and loop number $\loopnumber$, and its dual cochain complex has the structure of a differential graded Lie algebra (DGLA).

The \emph{odd graph complex} $\cGC_3$ is then the (commutative) graph complex where an orientation on a graph is specified by an ordering of the vertices and giving a direction to each edge, or equivalently~\cite[Proposition 1]{conant_theorem_2003}, an ordering of the edges and of a chosen cycle basis.
It is taken with the degree grading $k = \abs{E_{\Graph}} - 3\loopnumber$.
Due to the relations imposed on oriented graphs, for example, the class of a graph with a self-loop vanishes in $\cGC_3$ as it has an automorphism that only reverses the direction of the self-loop.
One of the major questions in current research is to compute the homology $H_{k}(\cGC_3)$, that is, to find linear combinations $\Graph$ of oriented graphs of degree $k$ such that $\partial G=0$ and there is no $F$ with $\partial F=G$.
Of what is currently known about $H^{\bullet}(\cGC_3)$ (see~\cite{willwacher_kontsevichs_2015,khoroshkin_differentials_2017,brun_graph_2024}), we briefly highlight the cohomology group $H^{-3}(\cGC_3)$. In particular, it forms a commutative algebra~\cite{duzhin_algebra_1998} and many non-trivial classes are known to exist due to its role in the study of knot invariants~\cite{bar-natan_vassiliev_1995}.

The perspective taken in~\cite{brown_unstable_2024} is to view integrals of the form $ I\defas \int \phi_{\Graph} \wedge \omega_{\Graph}$ as linear functionals on the odd graph complex, where $\omega$ is a \emph{canonical form}~\cite[Section~2.6]{brown_invariant_2021}, a differential form that is invariant under the action of the general linear group $\GL_n(\Z)$. That is, for a fixed choice of $\omega$, the integral $I$ assigns a transcendental number to each linear combination $G$ of oriented graphs.

Under certain conditions, these integrals define cocycles, i.e. $\delta I=0$, where $(\delta I)_F=I_{\partial F}=\sum_e I _{\Graph/e}$.
The idea is to use this property to check, for a closed linear combination $G$ of graphs, whether there exists a $F$ such that $G=\partial F$: If such $F$ exists, the integral $I_G= I_{\partial F}$ evaluates to zero.
At the same time, by pairing the integrals with such $G$, we can detect non-trivial \emph{co}homology classes $H^k(\cGC_3)$ of the odd graph complex.

In the simplest case one sets $\omega = 1$ and the integral is over the Pfaffian form by itself. $\int \phi_{\Graph}$, when restricted to graphs with loop number $\loopnumber = 2$, is a well-defined cocycle.
As it pairs non-trivially with the \emph{theta graph} (\cref{ex:pff-dipole}), which is closed, this cocycle thus corresponds to the only non-trivial class in the $\loopnumber=2$ graded piece of $H^{-3}(\cGC_3)$ as it is known to be of dimension $1$. In particular, it is dual to the class represented by the theta graph in homology.

The crucial properties of $\phi_{\Graph}$ that enabled their use in studying the cohomology of the odd graph complex is summarized by the following.
\begin{lemma}[Lemma 2.2, Theorem 6.7 of \cite{brown_unstable_2024}]\label{lem:pff-int-props}
  The Pfaffian form $\phi_{\Graph}$ for a graph $\Graph$ with $\loopnumber$ loops is a smooth, closed,  projective differential form on the open simplex of positive Schwinger parameters
  \begin{equation}\label{eq:simplex}
  \sigma_{\Graph} \defas \left\{ [a_{e_1}:\cdots : a_{e_{\abs{E_{\Graph}}}}]\ \text{with all}\  a_e>0 \right\} \subset \mathbbm{P}(\R^{E_{\Graph}}),
  \end{equation}
and transforms under the action of $P \in \GL_{\loopnumber}(\Z)$ via
\begin{equation}\label{eq:pff-change-basis}
  \phi_{\cyclelaplacian_{\cyclebasis'}}
    = \phi_{P^{\Transpose}\cyclelaplacian_{\cyclebasis}P}
    = \phi_{\cyclelaplacian_{\cyclebasis}} \cdot \det(P)
    = \pm \phi_{\cyclelaplacian_{\cyclebasis}}
\end{equation}
  for any two cycle bases $\cyclebasis' = \cyclebasis \cdot P$.
 Furthermore, its integral $I_{\Graph} = \int_{\sigma_{\Graph}} \phi_G$ is absolutely convergent.
\end{lemma}

Several properties of the Pfaffian form $\phi_\Graph$ have been established in \cite{brown_unstable_2024}, among them:
\begin{lemma}[Lemma 2.2,~Lemma 5.6,~Corollary 5.7 of \cite{brown_unstable_2024}] \label{lem:phi_properties}
	~\\[-.5cm]
	\begin{enumerate}[ref=\thelemma(\roman*)]
		\item $\phi_\Graph \wedge \phi_\Graph=0$,
		\item\label[lemma]{lem:odd-loop} $\phi_\Graph=0$ if $\Graph$, or any of its 1PI components, has odd loop number,
    \item\label[lemma]{lem:2-valent} $\phi_\Graph$ is compatible with subdividing edges. More precisely, let $\Graph'$ be the graph obtained from $\Graph$ by subdividing an edge $e$ into edges $e'$ and $e''$. This operation is realized by a map $s_e \colon \R^{E_{\Graph'}}_+ \to \R^{E_{\Graph}}_+, a=s_e(a')$ on Schwinger parameters $a'$, where   $a_e=s_e(a'_{e'}+a'_{e''})$ and $a_f = s_e(a'_f)$ for $f\neq e$.  Then
      $\phi_{\Graph'} = s_e^{*}\left( \phi_{\Graph}\right)$.
      Here we are taking the cycle basis for $ \Graph'$ induced from $\Graph$ by replacing any occurrence of $e$ by the path $e' e''$.
	\end{enumerate}
\end{lemma}

\subsection{Acknowledgements}
We thank Erik Panzer and Jingxiang Wu for multiple discussions regarding the broader context and interpretation of these constructions.
We also thank Francis Brown, Erik Panzer, Jingxiang Wu and Karen Yeats for valuable comments on earlier drafts.
SH is also grateful to Karen Yeats and the hospitality of Perimeter Institute, with whom and where some of this work was initiated.

PHB is funded through the Royal Society grant {URF{\textbackslash}R1{\textbackslash}201473}.
SH acknowledges the support of the Natural Sciences and Engineering Research Council of Canada (NSERC), [funding reference number 578060-2023]. 

\section{Discussion}\label{sec:consequences}

\subsection{Immediate algebraic properties}

\begin{corollary}
  By \cref{thm:main_thm}, the statements of \cref{lem:pff-int-props,lem:phi_properties} also hold for the topological form $\alpha_\Graph$, in particular
	\begin{enumerate}[ref=\thecorollary(\roman*)]
		\item\label[corollary]{lem:alpha_finite} $\alpha_\Graph$ is smooth and projective, and the integral $\int \alpha_\Graph$ is finite.
		\item \label[corollary]{lem:alpha_stokes}$\alpha_\Graph$ is closed and the integral $\int \alpha_\Graph$ satisfies Stokes relations.
		\item\label[corollary]{lem:alpha_2-valent} $\alpha_\Graph$ is compatible with subdivision of edges as in \cref{lem:2-valent}. More precisely, if $\Graph'$ is the graph obtained by subdividing an edge $e$, with labels and directions as specified in \cref{lem:subdivision_sign}, then $\alpha_{\Graph'} = (-1)^{e+1} s_e^*(\alpha_{\Graph})$.
		\item\label[corollary]{lem:alpha_square} $\alpha_\Graph \wedge \alpha_\Graph=0$ unless $\Graph$ is a tree.
	\end{enumerate}
\end{corollary}


Convergence of $\alpha_\Graph$ (\cref{lem:alpha_finite}) is clear in the operatope-formulation \cite{gaiotto_higher_2024} in terms of the variables $s_e$ (\cref{def:se}), but much less obvious from the parametric integrand in \cref{lem:dbarT}.
Via this connection to the Pfaffian form, the constructions in~\cite[Section 6.1]{brown_unstable_2024} explicitly realize the finiteness of $\int \alpha_{\Graph}$ as an integral over positive Schwinger parameters, as hinted at in~\cite[Sections 3.4, 3.8]{gaiotto_higher_2024}.
As we will explain further in \cref{sec:quadratic_relations},  quadratic relations between the integrals essentially correspond to Stokes relations (\cref{lem:alpha_stokes}).

\Cref{lem:alpha_2-valent} highlights one of the graph-related properties which is apparent when viewed through the Pfaffian form, but not obvious in the $\alpha_{\Graph}$ representation.
It is thus a shortcut for computing $\alpha_\Graph$ when $\Graph$ contains a 2-valent vertex. Moreover, it implies that  the integral $\int \alpha_\Graph$ vanishes for such graphs  because the resulting $\alpha_{\Graph'}$ has the same form degree, but one more integration variable, compared to $\alpha_{\Graph}$.

\begin{example}{}
  As an application of \cref{lem:alpha_2-valent}, we can use the Pfaffian form $\phi_{D_3}$ for the theta graph (\cref{ex:pff-dipole}) to compare with the topological forms for the dunce's cap $G'$ and for the double triangle $G''$, as computed in~\cite[Example 13]{balduf_combinatorial_2024} and~\cite[Eq. (3.56)]{gaiotto_higher_2024} respectively.

  \begin{center}
		\begin{tikzpicture}
		\node[vertex, label=below:$v_1$](v1) at (1,1){};
		\node[vertex, label=below:$v_2$ ](v2) at (-.5,-.5){};
		\node[vertex,label=below:$v_3$ ](v3) at (2.5,-.5){};

		\draw[edge, >=Triangle, decoration={markings, mark=at position .85 with {\arrow{>}}}, postaction={decorate}, ](v2) to node[pos=.5, above left]{$a'_1$}   (v1) ;
		\draw[edge, >=Triangle, decoration={markings, mark=at position .25 with {\arrow{<}}}, postaction={decorate}, ](v3) to node[pos=.5, above right]{$a'_2$}(v1);
		\draw[edge, >=Triangle, decoration={markings, mark=at position .25 with {\arrow{<}}}, postaction={decorate},  bend angle=30,bend left](v3) to node[pos=.5, below]{$a'_4$}(v2);
		\draw[edge, >=Triangle, decoration={markings, mark=at position .25 with {\arrow{<}}}, postaction={decorate},   bend angle=20, bend right](v3) to node[pos=.5, below]{$a'_3$}(v2);

		 \draw[line width=.3mm, |->] (4.0,0) --(3.4,0);

		\node[vertex,label=below:$v_1$](v2) at (4.9,0){};
		\node[vertex,label=below:$v_2$](v3) at (7.7,0){};

		\draw[edge, >=Triangle, decoration={markings, mark=at position .25 with {\arrow{<}}}, postaction={decorate},  bend angle=50,bend right](v3) to node[pos=.5, below]{$a_1$}(v2);
		\draw[edge, >=Triangle, decoration={markings, mark=at position .25 with {\arrow{<}}}, postaction={decorate},  bend angle=0,bend left](v3) to node[pos=.5, below]{$a_2$}(v2);
		\draw[edge, >=Triangle, decoration={markings, mark=at position .25 with {\arrow{<}}}, postaction={decorate},   bend angle=50, bend left](v3) to node[pos=.5, below]{$a_3$}(v2);

		\draw[line width=.3mm, |->] (8.6,0) --(9.2,0);

		\node[vertex, label=below:$v_1$](v1) at (11.5,1){};
		\node[vertex, label=below:$v_3$](v2) at (10,0){};
		\node[vertex, label=below:$v_4$](v3) at (13,0){};
		\node[vertex, label=below:$v_2$](v4) at (11.5,-1){};

		\draw[edge, >=Triangle, decoration={markings, mark=at position .25 with {\arrow{<}}}, postaction={decorate},](v1) to node[pos=.6, label={[label distance=-.5]above:{$t_{12}$}}]{}   (v2) ;

		\draw[edge, >=Triangle, decoration={markings, mark=at position .85 with {\arrow{>}}}, postaction={decorate},](v1) to node[pos=.6,label={[label distance=-2]above:{$t_{31}$}}]{}   (v3) ;

		\draw[edge, >=Triangle, decoration={markings, mark=at position .85 with {\arrow{>}}}, postaction={decorate},](v2) to node[pos=.5, below]{$t_{23}$}   (v3) ;

		\draw[edge, >=Triangle, decoration={markings, mark=at position .85 with {\arrow{>}}}, postaction={decorate},](v2) to node[pos=.4, label={[label distance=-.7]below:{$t_{24}$}}]{}   (v4) ;
		\draw[edge, >=Triangle, decoration={markings, mark=at position .25 with {\arrow{<}}}, postaction={decorate},](v3) to node[pos=.4, label={[label distance=0]below:{$t_{43}$}}]{}   (v4) ;

		\node at (1,-2){$G'$};
		\node at (6.3,-2){$D_3$};
		\node at (11.5,-2){$G''$};

	\end{tikzpicture}
	\end{center}

  Notice that the graphs have different numbers of spanning trees, as well as different denominators in the expansion of their topological forms via \cref{lem:dbarT}.
	Hence, the observation that these differential forms coincide up to changes of variables is a structural property of $\alpha_\Graph$ resp. $\phi_\Graph$.
  Explicitly, the transformations for the two cases are:
  \begin{align*}
    \frac{a_3 \d a_1 \wedge \d a_2}{(a_1a_2 + a_2a_3 + a_1a_3)^{3/2}}
      &\mapsto \frac{a'_4 (\d a'_1 + \d a'_2) \wedge \d a'_3}{\big((a'_1 +a'_2)a'_3 + (a'_1 + a'_2)a'_4 + a'_3 a'_4\big)^{3/2}} \\
      &\mapsto \frac{(t_{24} + t_{43}) (\d t_{12} + \d t_{31}) \wedge \d t_{23}}{\big(  (t_{12} + t_{31}) t_{23} + (t_{12}+ t_{31})(t_{24}+ t_{43}) + t_{23}(t_{24}+t_{43}) \big)^{3/2}},
  \end{align*}
  for the terms corresponding to the spanning tree contributions of $T=\{3\}$, $T' \in \{\{1,4\},\{2,4\}\}$, and $T'' \in \{\{12,24,43\}, \{31,24,43\}\}$ respectively (\cref{ex:dunces_cap_tree} computes the $T'=\{2,4\}$ term).

  The first line, splitting the edge $e=1$ into $a_1 = a'_1 + a'_2$ and relabelling $a_2 = a'_3, a_3 = a'_4$, produces the dunce's cap $G'$. By \cref{lem:alpha_2-valent}, the relative sign is $(-1)^{1+1}=+1$ for the transformation from $\alpha_{D_3}$ to $\alpha_{G'}$, and the result coincides with \cref{ex:dunces_cap}. If we work with Pfaffian forms $\phi_{G'}= s_1^*(\phi_{D_3})$, given by \cref{eq:pff-D3}, there is a factor $(-4)$ from \cref{thm:main_thm} to relate them to the topological form, so that overall  $\alpha_{G'} = (-4)^{-1} \phi_{G'}= (-4)^{-1}s_1^*(\phi_{D_3})$.  
  
  The topological form $\alpha_{\Graph''}$ for the double triangle is
	\begin{align*}
	\scalemath{.85}{\alpha_{G''}  = -\frac{(t_{12}+ t_{31}) \d t_{23} \wedge (\d t_{24} + \d t_{43}) - t_{23}(\d t_{12} + \d t_{31})\wedge (\d t_{24} + \d t_{43}) + (t_{24} + t_{43}) (\d t_{12} + \d t_{31}) \wedge \d t_{23}  }{8 \pi  \big(  (t_{12} + t_{31}) t_{23} + (t_{12}+ t_{31})(t_{24}+ t_{43}) + t_{23}(t_{24}+t_{43}) \big)^{3/2}  }.}
	\end{align*}
    \Cref{lem:alpha_2-valent} with the change of variables $a_1 = t_{12} + t_{31}$, $a_2 = t_{23}$ and $a_3 = t_{24} + t_{43}$ gives
  $\alpha_{G''} = (+1) s_1^*\big((+1) s_3^*(\alpha_{D_3})\big)$.
  We can then compare with $\phi_{D_3}$ from \cref{ex:pff-dipole}, which verifies that indeed
  $\alpha_{G''} = (-4)^{-1} s_{1}^*\big(s_{3}^*(\phi_{D_3})\big),$
  as predicted by \cref{thm:main_thm}.

  Note that our  expression $\alpha_{G''}$ has an overall minus compared to \cite[Eq.~(3.56)]{gaiotto_higher_2024} because in the latter case, the vertex $v_1$ is used as the special vertex $v_\star$, while \cref{lem:dbarT} uses the highest vertex, $v_\star=v_4$, and we have chosen a different labelling in order to directly apply \cref{lem:alpha_2-valent}. When comparing, note that the drawing in \cite{gaiotto_higher_2024} shows a different labelling of the double triangle graph than what is used in the computation.
\end{example}

\Cref{lem:alpha_square} had been proved for $\alpha_\Graph$, with some effort, in \cite{balduf_combinatorial_2024,wang_factorization_2024}, whereas this result for the Pfaffian form $\phi_\Graph$  is a straightforward consequence of linear algebra in \cite[Lemma 2.2]{brown_unstable_2024}.
To understand its significance  in the topological case, recall that $\int \alpha_\Graph$ computes the \emph{violation} of BRST-closedness of observables. $\alpha_\Graph \wedge \alpha_\Graph=0$ therefore implies that a theory which has two topological directions, and \emph{arbitrarily many} more topological or holomorphic directions, is free of quantum BRST anomalies. The special case of just two topological directions amounts to integrals of the form
\begin{align*}
I_\Graph &= \int_{\left \lbrace a_e \right \rbrace  }\alpha_\Graph \wedge \alpha_\Graph =  \int_{\left \lbrace a_e \right \rbrace   } \int_{\left \lbrace x^{(1)}_v \right \rbrace } \int_{ \left \lbrace x^{(2)}_v \right \rbrace   } W^{(1)}_\Graph \wedge W^{(2)}_\Graph,
\end{align*}
where $x_v= (x^{(1)}_v, x^{(2)}_v)$ is a 2-component real vector, and $W_\Graph^{(i)} = \bigwedge_e e^{-{s^{(i)}_e}^2} \d s^{(i)}_e$. If, instead of solving the position integrals first, we were to compute the $\alpha_e$ integrals first, the integrand would be a product of terms
\begin{align*}
e^{-{s_e^{(1)}}^2 - {s_e^{(2)}}^2 } \d s^{(1)}_e \wedge \d s^{(2)}_e = e^{-\frac{\abs{\vec x}^2}{a}} \left( -2 a_e^{-2} \d a_e \left( x^{(2)} \d x^{(1)} - x^{(1)} \d x^{(2)} \right) + a_e^{-1} \d x^{(1)} \wedge \d x^{(2)} \right) .
\end{align*}
The last term in the parenthesis does not contribute to an integral over $a_e$, and the first term  equals the angle differential in 2-dimensional polar coordinates, $ x^{(1)} \d x^{(2)} - x^{(2)} \d x^{(1)} = \left( (- x^{(2)})^2 + (x^{(1)})^2 \right) \d \phi = \abs{\vec x}^2 \d \phi$. Consequently, the parametric integral of a single edge is
\begin{align*}
	\int_{a_e=0}^\infty 	e^{- {s^{(1)}}^2 - {s^{(2)}}^2} \d s^{(1)} \wedge \d s^{(2)} &= 	\int_{a_e=0}^\infty e^{-\frac{\abs{\vec x}^2}{a_e}} 2 a_e^{-2}  \abs{\vec x}^2 \d \phi_e \wedge \d a_e = 2 \d \phi_e .
\end{align*}
The remaining position-space integral is therefore over the angle of edge $e$ in the plane. These are the integrals whose vanishing, proved in \cite{kontsevich_deformation_2003}, asserts the absence of obstructions to deformation quantization of Poisson manifolds (\emph{Formality theorem} \cite{kontsevich_formality_1997}).

\medskip

We want to briefly mention that in \cite[Sec.~5]{kontsevich_formality_1997}, Kontsevich remarked that \emph{failure} of the formality theorem would imply the non-vanishing of the cohomology group   $H^1(\cGC_2)$ in the \emph{even graph complex}\footnote{What is called \emph{even} graph complex in our notation was called \emph{odd} graph complex in \cite{kontsevich_formality_1997}.}. Conversely, the formality theorem does not imply the vanishing,   but still,  $H^1(\cGC_2)$ is conjectured to be trivial. However, the Pfaffian form $\phi_\Graph$ (\cref{def:pff-form}) is being used on the \emph{odd} graph complex, therefore our result that $\alpha_{\Graph}=\phi_\Graph$ is not directly related to this question.

\subsection{Dipole / multiedge graphs and the Moyal product}

On both sides of the equivalence, the differential forms $\alpha_{\Graph}$ resp. $\phi_\Graph$ are of interest for arbitrary connected graphs because we can take exterior products with further differential forms.

When $\alpha_{\Graph}$ or $\phi_\Graph$ are considered by themselves, their form degree implies that their integral is only non-vanishing for connected graphs with two vertices and $\loopnumber+1$ edges. As \cref{lem:odd-loop} implies that these forms vanish if $\Graph$ has a self-loop, the only such graphs are the \emph{dipole} graphs. 
Let   $D_{2i+1}$ be the dipole graph with $2i+1$ edges, all oriented from $v_0$ to $v_1$, it has $2i$ loops.
The case when $i=1$ is given in~\cref{ex:pff-dipole}.
Its Pfaffian form is~\cite[Corollary~5.7]{brown_unstable_2024},
\begin{align}\label{dipole_pfaffian}
 \phi_{D_{2i+1}} &=  \frac{(2i)!}{(4\pi)^i\cdot i!}\frac{(a_1 \cdots a_{2i+1})^{i-1}}{\firstsymanzik^{i+\frac 12}} \sum_{e=1}^{2i+1} (-1)^{e-1} a_e \d a_1 \cdots \wedge \widehat {\d a_e} \wedge \cdots \wedge \d a_{2i+1}.
\end{align}
and its parametric integral was solved explicitly in \cite[Example~6.2]{brown_unstable_2024},
\begin{align}\label{dipole_integral}
I_{D_{2i+1}} = \int \phi_{D_{2i+1}} = 1 \qquad \forall i.
\end{align}
Conversely, the integral of the topological form $\alpha_\Graph$ has been computed in its position-space formulation in \cite[Eq.~(3.28)]{gaiotto_higher_2024}. The result $I_{D_{2i+1}} = 2^{-2i} = 2^{-\loopnumber}$ is compatible with \cref{thm:main_thm}, $\alpha_\Graph = 2^{-\loopnumber} \phi_\Graph$, where the relative sign is $+1$ for this choice of edge directions (\cite[Example 4.18]{brown_unstable_2024}).

\medskip
In the setting of topological field theory, the sum over all dipole (Feynman-) integrals amounts to  the two-argument bracket (\cref{bracket_integral}). This setup is discussed in \cite[Sec.~3.7.1 and 4.7.1]{gaiotto_higher_2024} for the example of   1-particle bosonic topological quantum mechanics. In that case, the \enquote{fields} are the quantum mechanical operators $q,p$, which in ordinary quantum mechanics (in the Heisenberg picture) would be functions of time, $q(t),p(t)$. In a  topological theory, they are independent of time. Hence, the observables $\mc O$ are monomials in $p$ and $q$, taken at any fixed time.
 The Feynman graphs $I_\Graph$ are built by Wick-contracting two such monomials. At tree-level, there is a single propagator, connecting one factor of $p$ to one factor of $q$, and all other factors remain. Since Feynman rules imply a sum over all choices to pick the particular $p$ out of $\mc O_1$ and $q$ out of $\mc O_2$, the combinatorial prefactor is produced correctly by the derivative $\partial_p \mc O_1$ and $\partial_q \mc O_2$ (without setting $p=0$ nor $q=0$). Hence, the tree-level term corresponds to
\begin{align}\label{Moyal_1}
	\begin{tikzpicture}[baseline=-.1cm]
	\node[vertex,label=below:$\mc O_1$](v1) at (0,0) {};
	\node[vertex,label=below:$\mc O_2$](v2) at (1.5,0) {};
	\draw[edge] (v1) -- (v2);
\end{tikzpicture} &=\left( \partial_p \mc O_1 \right) \left( \partial_q \mc O_2 \right) -  \left( \partial_q \mc O_1 \right) \left( \partial_p \mc O_2 \right)=: \eta^{ij} (\partial_i \mc O_1)(\partial_j \mc O_2), \quad \eta\defas \begin{pmatrix}
0 & -1 \\ 1 & 0
\end{pmatrix}.
\end{align}
We have introduced the $2\times 2$ matrix $\eta^{ij}=-\eta^{ji}$ and vectors $(\partial_q, \partial_p)^\Transpose$, summation over $i$ and $j$ is implied. By \cref{bracket_integral}, this expression is weighted by the integral $\frac{I_\Graph}{\abs{\Aut(\Graph)}}$, which is   unity since $\int \alpha_\Graph=\pm 1$ for all tree graphs.

The next term in the series vanishes since it has odd loop order. The third term is given by the two-loop dipole integral (\cref{ex:pff-dipole}), with $I_{D_3}= 2^{-2}$ and $\abs{\Aut(D_3)}= 3!$. It amounts to a sum over all ways of Wick-contracting three pairs of $(p,q)$ variables.
\begin{center}
	\begin{tikzpicture}
		\node[vertex,label=below:$\mc O_1$](v1) at (0,0) {};
		\node[vertex,label=below:$\mc O_2$](v2) at (1.5,0) {};
		\draw[edge, bend angle=30, bend left] (v1) to (v2);
		\draw[edge, bend angle=30, bend right] (v1) to (v2);
		\draw[edge] (v1) -- (v2);
		\node at (2.2,-.1){$=$};
		\node[font=\small] at (8.7,-.1){$\left( \partial^3_p \mc O_1 \right) \left( \partial^3_q \mc O_2 \right) -  \left( \partial^2_q \partial_p \mc O_1 \right) \left( \partial^2_p \partial_q \mc O_2 \right)\pm \ldots = \eta^{ij} \eta^{kl}\eta^{mn} (\partial_i \partial_k \partial_m \mc O_1)(\partial_j \partial_l \partial_n \mc O_2).  $};
	\end{tikzpicture}
\end{center}
Continuing this way, we obtain a sum of the form
\begin{align}\label{moyal_commutator}
\left \lbrace \mc O_1, \mc O_2 \right \rbrace &= \sum_{n=0}^\infty \frac{1}{4^n} \frac{1}{(2n+1)!} \left( \eta^{ij} \right) ^{2n+1} \left( \partial^{2n+1} \mc O_1 \right) \left( \partial^{2n+1} \mc O_2 \right)  ,
\end{align}
which is, by inspection of the terms, identified as the Moyal commutator between $\mc O_1$ and $\mc O_2$,
\begin{align*}
\left \lbrace \mc O_1, \mc O_2 \right \rbrace &= \mc O_1 \star \mc O_2 - \mc O_2 \star \mc O_1=\left[ \mathcal O_1, \mathcal O_2 \right] _\star.
\end{align*}
In view of BRST cohomology (\cref{sec:intro_alpha}), the bracket expresses the violation of BRST-closedness by arbitrary operators $\mc O_1, \mc O_2$. If one adds an interaction monomial $\mc I$ to the action, then, setting   $\mc O_1=\mc I$, the choice $Q_\text{quantum}=Q_\text{free}- \left \lbrace \mc I, \cdot \right \rbrace $ would be the corresponding BRST differential at quantum level. In particular, all higher-order brackets $\left \lbrace \mc I, \mc I, \cdot \right \rbrace $ vanish.

\bigskip
In the world of graph complexes,  the analogous sum over dipoles weighted by automorphism factors\footnote{Here and in~\cite{brown_unstable_2024}, we are viewing graphs as defined via half-edges and thus the dipole graphs (including the single edge $D_1$) have an extra automorphism that swaps the two vertices i.e. $\abs{\Aut(D_{2i+1})} = 2(2i+1)!$} plays a special role, too.  Recall that one interprets an integral $I_\Graph$ as an element of the dual.
The fact that the integral of $\phi_\Graph$ is non-zero only for dipole graphs means that it can be viewed as the pairing~\cite[Corollary 6.3]{brown_unstable_2024}
\begin{align}\label{eq:m-dipole}
  I_{\Graph} = \inner{\Graph, \mathfrak{m}}
  \qquad\text{with the dipole sum }\qquad \mathfrak{m} \defas \sum_{i=1}^{\infty} \frac{D_{2i+1}}{2 (2i+1)!}.
\end{align}
Here, the pairing $\left \langle \cdot, \cdot \right \rangle $ is defined such that it identifies the chain complex $\cGC_3$ with its dual cochain complex, i.e. the differential $\partial$ is dual to the codifferential $\delta$. On graphs, $\delta G$ is the sum of all graphs obtained by inserting an edge $e\in \Graph$ in turn (with appropriate signs).
The cochain complex has the structure of a DGLA, where the Lie bracket $[\cdot,\cdot]$ is defined via graph insertions. It is dual to the anti-symmetrization of the map\footnote{The reader might recognize this construction from the Hopf algebra theory of renormalization \cite{kreimer_hopf_1998}, where a similar-looking \emph{renormalization coproduct} describes how divergent subgraphs $\gamma$ are to be replaced by counterterms. However, for the latter, the sum goes over disjoint unions of 2-connected (=1PI) subgraphs $\gamma$, which in general is not equivalent to $\gamma$ being in $\cGC_3$ as requested in our case.}
\begin{align}\label{def:coproduct}
  \Delta \Graph \defas \sum_{\gamma} \gamma \otimes \Graph / \gamma,
\end{align}
where the sum runs over all subgraphs $\gamma$ where $\gamma$ and $\Graph/\gamma$ are both in $\cGC_3$.
On the level of chain complexes, the anti-symmetrization of $\Delta$ thus gives a map to the tensor product $\cGC_3 \otimes \cGC_3$ of complexes where the differential is given by $\partial \otimes 1 + 1 \otimes \partial$.

The pairing $\inner{\Graph, \mathfrak m}$ is only non-zero, and equal to $\pm 1$, when $\Graph$ is isomorphic to a dipole with even loop number.
As first discovered combinatorially in \cite{khoroshkin_differentials_2017}, $\mathfrak{m}$ is a Maurer-Cartan element for the odd graph complex and thus defines a \emph{twisted differential} $\delta + [\cdot, \mathfrak{m}]$ on $\cGC_3$.
 Insertion of a single edge is a special case of inserting a graph, therefore, by extending the dipole sum $\mathfrak{m}$ to include the single edge $D_1$ (i.e. including $i=0$), which we denote by $\mathfrak{m'}$, we can rewrite the twisted differential as a Lie bracket with $\mathfrak m'$,
\begin{align}\label{twisted_differential}
 \delta + [\cdot, \mathfrak{m}] = [\cdot , \mathfrak{m'}].
\end{align}
The cohomology of $\cGC_3$ with respect to \emph{this}  differential (as opposed to the usual differential $\delta$) was completely computed in~\cite[Theorem 2]{khoroshkin_differentials_2017}, with the only non-trivial class being a similar sum over dipoles, but weighted with a slightly different factor.
As we will explain in~\cref{sec:quadratic_relations}, this twisted differential arises naturally in the study of $I_{\Graph}$.

\subsection{Quadratic /  Stokes relations}\label{sec:quadratic_relations}

The \emph{operatope} $\Delta_\Graph$ introduced in \cite[Eq.~(2.30)]{budzik_feynman_2023} and \cite[Eq.~3.16]{gaiotto_higher_2024} is the integration domain of $\alpha_\Graph$ in terms of the variables $s_e$ (\cref{def:se}). If $S\subset V_\Graph$ is a subset of vertices, then denote by $G[S]\subset G$   the corresponding induced subgraph, and by $ \Graph/S$   the cograph, that is, the graph that is obtained when $G[S]$ is contracted to a single vertex in $\Graph$. The authors of \cite{gaiotto_higher_2024} observe that the holomorphic-topological integrals $I_\Graph$ satisfy \emph{quadratic identities}, which correspond to nilpotency of the BRST differential. Let $\Graph$ be a connected graph with $\abs{V_\Graph}=3$ vertices\footnote{\cite{gaiotto_higher_2024,budzik_feynman_2023} consider a more general case of $T$ topological and $H$ holomorphic directions, then the corresponding graphs are $(H+T)$-Laman \cite{pollaczek-geiringer_uber_1927,henneberg_graphische_1911,laman_graphs_1970}. In our case, $H=0$ and $T=1$, a Laman graph has $\abs{V_\Graph}=2$. }, then
\begin{align}\label{quadratic_identity}
	\sum_{ S \subset G, ~\abs{V_S}=2} \sgn(\Graph,S) \Delta_{\Graph[S]} \times \Delta_{\Graph/S} &=0.
\end{align}
In the present topological case, as the integral $I_\Graph=\int \alpha_\Graph$ is independent of kinematic data, this identity directly corresponds to an equivalent factorization of the integrals $I_{\Graph[S]} \times I_{\Graph/S}$ themselves. The connected graphs with $\abs{V_\Graph}=3$  are triangles with multiple edges. Their induced subgraphs are dipoles, so \cref{quadratic_identity} asserts that the integrals $I_{D_{2i+1}}$ of the $2i$-loop dipole graphs are related for different values of $i$.

On the other hand, for a suitable compactification of the simplex $\sigma_{\Graph}$ (\cref{eq:simplex}) called the \emph{Feynman polytope} $\widetilde{\sigma}_{\Graph}$~\cite{bloch_motives_2006, brown_feynman_2017}, the boundary components are in bijection with, and factorize into,
\[ \widetilde{\sigma}_{\gamma} \times \widetilde{\sigma}_{\Graph / \gamma }\]
for all bridgeless (1PI) subgraphs $\gamma$ of $\Graph$ and single edges $\gamma = \{e\}$ (in which case the polytope $\widetilde{\sigma}_{e}$ is a point).
Based on the closedness ($\d \phi_\Graph=0$) of the Pfaffian form and how it factorizes at the boundary of $\widetilde{\sigma}_{\Graph}$, Stokes' theorem gives \cite[Section 6.2]{brown_unstable_2024}
\begin{align}\label{Stokes_relation}
  0 &= \delta I_{\Graph} + \frac{1}{2} [I_{\Graph}, \mathfrak{m}] = \delta\mathfrak{m} + \frac{1}{2} [\mathfrak{m}, \mathfrak{m}],
\end{align}
which is equivalent to the Maurer-Cartan equation for the dipole sum $\mathfrak{m}$ of \cref{eq:m-dipole}.
That is, the twisted differential $\delta + [\cdot, \mathfrak{m}]$ of \cref{twisted_differential} manifests geometrically due to the boundary structure of $\widetilde{\sigma}_{\Graph}$.

We thereby see that in the case of the Pfaffian/topological form, the \emph{quadratic identity} for $I_{\Graph}$ resulting from \cref{quadratic_identity} for the operatope directly corresponds to the Stokes relation $0=\delta I_{\Graph} + \frac{1}{2} [I_{\Graph},I_{\Graph}]$. In both cases, this identity only gives non-trivial relations when applied to graphs with 3 vertices, and the relations in both cases lead to the identity $I_{D_{2i+1}} = \left( I_{D_3} \right) ^i$ (see \cite[Example 6.9]{brown_unstable_2024} and \cite[Sec.~3.7.1]{gaiotto_higher_2024}).

\subsection{Outlook}

The proof of the equivalence of the Pfaffian form $\phi_\Graph$ and the topological form $\alpha_\Graph$ given in the present work is an explicit combinatorial construction, based on their respective definitions.

As pointed out above, the BRST-brackets and the topological form  $ \alpha_\Graph$ are intimately related to the formality theorem \cite{kontsevich_deformation_2003} for topological field theory \cite{cattaneo_path_2000}, and the graph complex $\cGC_3$ where the form $\phi_\Graph$ operates, has originally been introduced in \cite{kontsevich_formal_1993} in a related context.
Conversely,  in the path integral formulation of Chern Simons theory, the sign computed by a Pfaffian of the field differential operator is crucial to relate the QFT picture to Donaldson invariants, as  discussed extensively in \cite[Section~3]{witten_topological_1988}.  Of course, a priori the latter Pfaffian is a totally different object from our Pfaffian form $\phi_\Graph$. Nevertheless, one can speculate whether  it would be possible to track the appearance of Pfaffians all the way from path integrals in TQFT to their series expansion in terms of Feynman diagrams to their role in the graph complex.

Secondly,  the sum of dipole Feynman integrals (\cref{moyal_commutator}) represents the only quantum correction $ \left \lbrace \mc I, \cdot  \right \rbrace $ to the BRST differential  for the case of topological quantum mechanics.  Moreover, it follows from the cohomology of the Poisson operad  that $d>1$-dimensional topological quantum field theory has no non-trivial higher brackets either \cite[Sect.~3.3]{beem_secondary_2020}. At the same time, the dipole sum $\mathfrak m$ defines a twisted differential in \cref{twisted_differential}, and the entire cohomology of $\cGC_3$ with respect to this differential is itself given by a sum of dipoles. It would be interesting to understand if there is any relation between these observations (cf.~\cite[Remark 2, Appendix A]{khoroshkin_differentials_2017}).

Thirdly,  the topological form $\alpha_\Graph$ arises from integrals where $Q_\text{free}$ acts on a product of propagators.
The free BV-BRST operator $Q_\text{free}$ satisfies \emph{descent relations} \cite{beem_secondary_2020,gaiotto_higher_2024} $Q_\text{free} \mc O = -\d \mc O$, where $\d$ is the free-field differential operator. On the other hand, the Feynman propagator $P(x)$ is by definition a solution of $\d P(x) = \pi \delta(x)$. The integrals $I_\Graph$ therefore consist of a sum of terms where one propagator in turn is \enquote{contracted}, up to regularization. Having made the explicit relation with the graph complex $\cGC_3$, a natural question is to what extent this contraction of propagators corresponds to the differential $\partial$ in the graph complex, which is defined as the signed sum of all ways to contract one edge of $\Graph$. One approach would be to   construct an explicit  position-space representation of $\alpha_{\Graph}$ that carefully takes into account regularization.

\medskip

Lastly, as alluded to earlier,
in the context of graph complexes one often wants to take products of  $\phi_{\Graph}$ with further differential forms $\omega_\Graph$ \cite{brown_unstable_2024}, whose integrals could then pair with graphs outside of the dipole family. Similarly, in~\cite{gaiotto_higher_2024}, a second class of differential forms, $\rho_\Graph$, has been introduced to study  the analogous corrections in theories with holomorphic directions.
In an upcoming article, we will give a description of the holomorphic form $\rho_\Graph$ in terms of graph matrices.

\section{Definitions and proofs}\label{sec:definitions_proofs}

The remainder of the paper is devoted to explaining all graph-theoretic notions that appear in the formulas of Pfaffian- and topological differential forms, and to establishing the lemmas that ultimately lead to a proof of \cref{thm:main_thm}.

\subsection{Incidence and path matrices} \label{sec:incidence_path_matrices}

We first recall some standard definitions of graph theory. We will often have to fix one special vertex $v_\star$. This choice is arbitrary, but we usually take $v_\star = v_{\abs{V_\Graph}}$ in order to have a natural labelling on the remaining vertices $\overline V_\Graph \defas V_\Graph \setminus \left \lbrace v_\star \right \rbrace $.
\begin{definition}\label{def:vertex_incidence_matrix}
	The reduced \emph{vertex incidence matrix} $\incidencematrix$ of $\Graph$ is a $\abs{E_\Graph}\times (\abs{V_\Graph}-1)$--matrix where rows are indexed by edges $e\in E_\Graph$ and the columns are indexed by vertices $v\in \overline V_\Graph$, and the entries are
	\begin{align*}
	\left( \incidencematrix \right) _{e,v} \defas \begin{cases}
		-1 & \text{if $e$ starts at $v$, that is, tail$(e) = v$},\\
		1 & \text{if $e$ ends at $v$, that is, head$(e) = v$},\\
		0 &\text{else}.
	\end{cases}
	\end{align*}
\end{definition}

 We introduce one formal parameter $a_e$ for every edge $e$. In physics, these are the \emph{Schwinger parameters}. From these parameters, we define the \emph{edge variable matrix}
\begin{align}\label{def:edge_variable_matrix}
	\edgematrix \defas \diag(\vec a)=\begin{pmatrix}
		a_1 & \\
		& \ddots & \\
		&& a_{\abs{E_\Graph}}
	\end{pmatrix}.
\end{align}
Powers of the matrix $\edgematrix$ are understood in terms of 
the corresponding power of the diagonal entries. In particular, we often use
\begin{align*}
	\edgematrix^{-1} = \diag\left(  \vec {a^{-1}} \right).
\end{align*}

A \emph{path} from $v$ to $w$ is a connected sequence of edges that does not visit any vertex more than once.
A \emph{cycle} is a closed path, i.e. a path where $v=w$.
 The number of linearly independent cycles in a graph $\Graph$ is, by Euler's formula,
\begin{align}\label{def:loopnumber}
\loopnumber=\abs{E_\Graph}-\abs{V_\Graph}+1,
\end{align}
it is called the \emph{loop number} of the graph.
As a cycle can be identified with the set of edges it contains, we can denote them by a column vector $C_j$ with $\abs{E_\Graph}$ entries, where the entry $(C_j)_k$ is $(+1)$ if the cycle $j$ passes through edge $k$ in positive direction, $(-1)$ if cycle $j$ passes through edge $k$ in negative direction, and zero otherwise.
The set of cycles form a vector space over $\Z$ of dimension $\loopnumber$ called \emph{cycle space} $\cyclespace$ and any set of $\loopnumber$ independent cycles forms a basis of $\cyclespace$.

\begin{definition}\label{def:cycle_incidence_matrix}
	Let $\Graph$ be a connected graph and let $ \{C_1, \ldots, C_{\loopnumber}\}$ be an ordered cycle basis of $\Graph$. 
  The corresponding \emph{cycle incidence matrix} $\cycleincidencematrix$ is an $\abs{E_\Graph}\times \loopnumber$--matrix whose columns are these edge vectors $C_j$.
\end{definition}

A \emph{tree} is a connected graph without cycles. A \emph{spanning tree} $T \subseteq \Graph$ is a tree that is incident to all vertices $V_\Graph$ of $\Graph$. A choice of spanning tree $T \subseteq \Graph$ induces a choice of cycle basis since, by \cref{def:loopnumber}, there are exactly $\loopnumber$ edges $e\in E_\Graph$ which are not in $T$. Adding one of these edges to $T$ forms a cycle, and these cycles are necessarily linearly independent since the defining edge $e$ is only contained in the cycle it defines. The so-obtained set of cycles is called the \emph{fundamental cycle basis} induced by $T$.
\begin{definition}\label{def:fundamentalcyclebasis} We denote by $\fundamentalcbasis$ the cycle incidence matrix (\cref{def:cycle_incidence_matrix}) corresponding to a fundamental cycle basis induced by a spanning tree $T$.
\end{definition}

A graph typically allows for several distinct choices of paths between two given vertices. Choosing one path to each vertex gives rise to a \emph{path matrix}:

\begin{definition}\label{def:pathmatrix}
	A \emph{path matrix} $\pathmatrix$ is a $\abs{E_\Graph} \times (\abs{V_\Graph}-1)$--matrix where the column $j$ represents a (directed) path from $v_\star$ to $v_j$. That is, $\pathmatrix_{e,j}=+1$ if the edge $e$ appears in the path in its natural direction, $-1$ when it is used in reverse direction, and $\pathmatrix_{e,j}=0$ when $e$ is not part of the path $j$.
\end{definition}

The columns of a path matrix are linearly independent, and they span a subspace of edge space of dimension $\abs{V_\Graph}-1$, and the columns of $\pathmatrix$ are linearly independent of the columns of $\cycleincidencematrix$.

\begin{lemma}\label{lem:pathmatrix_incidencematrix_orthogonal}
	Let $\incidencematrix$ be the vertex incidence matrix (\cref{def:vertex_incidence_matrix}), $\cycleincidencematrix$ any choice of cycle incidence matrix (\cref{def:cycle_incidence_matrix}), and $\pathmatrix$ any choice of path matrix. Then
	\begin{enumerate}
		\item $\pathmatrix^\Transpose \incidencematrix =\identitymatrix_{(\abs{V_\Graph}-1)\times (\abs{V_\Graph}-1)}$,
		\item $	\incidencematrix^\Transpose \cycleincidencematrix = \zeromatrix_{(\abs{V_\Graph}-1) \times \loopnumber},\quad \text{and equivalently} \qquad
		\cycleincidencematrix^\Transpose \incidencematrix = \zeromatrix_{\loopnumber \times (\abs{V_\Graph}-1)} $.
	\end{enumerate}
\end{lemma}
\begin{proof}
	We only show point 1, the second point follows from an analogous argument.

	The matrix product $\pathmatrix^\Transpose \incidencematrix$ is the dot product between columns of $\incidencematrix$ and rows of $\pathmatrix^\Transpose$. Column $i$ of $\incidencematrix$ consists of the edges incident to vertex $i$. Row $j$ of $\pathmatrix^\Transpose$ is column $j$ of $\pathmatrix$, it consists of edges that form a path from $v_\star$ to $v_j$.

	Notice that the path to $v_j$ contains exactly one edge adjacent to $v_j$, with sign $+1$ if the edge points towards $v_j$. If $i=j$, this edge is the only one that simultaneously appears in column $i$ of $\incidencematrix$, and it appears with the same sign. Therefore, the diagonal entries of the matrix product are $+1$.
	If $i\neq j$, there are two possibilities: Either none, or exactly two, of the edges in the path to $v_j$ is contained in column $i$. In the first case, the dot product vanishes. In the second case, the two edges appear with opposite relative signs, and the sum is zero.
\end{proof}

\subsection{Laplacians}

We will often  consider  minors of graph matrices.
Let $M$ be an $n \times m$-matrix and consider sets $A$ and $B$ corresponding to subsets of rows and columns of $M$ respectively.
We define two types of submatrices:
\begin{align}\label{def:mat-minors}
	M( A,B ) &\defas M \text{ where rows $A$ and columns $B$ have been removed}\\
  M[ A,B ] &\defas M \text{ where only rows $A$ and columns $B$ are present}. \nonumber
\end{align}
If $\abs{A}=\abs{B}$, then these matrices are square and it makes sense to compute their determinant.
We use \enquote{$-$} in place of $A$ or $B$ to denote all rows or all columns of $M$ and use the shorthand
\begin{align}\label{submatrix_rows}
M[A] \defas M[A, -]
\end{align}
for the submatrix obtained by taking the rows of $M$ corresponding to $A$.

We need two different notions of Laplacian matrices. In both cases, they are \emph{reduced} in the sense that they do not involve the removed vertex $v_\star$.
\begin{definition}\label{def:Laplacian}
	The \emph{expanded vertex Laplacian} is a $(\abs{E_\Graph}+\abs{V_\Graph}-1)\times (\abs{E_\Graph}+\abs{V_\Graph}-1)$--matrix consisting of the reduced incidence matrix $\incidencematrix$ (\cref{def:vertex_incidence_matrix}) and the edge matrix (\cref{def:edge_variable_matrix}):
	\begin{align*}
		\explaplacian &\defas \begin{pmatrix}
			\edgematrix & \incidencematrix \\
			-\incidencematrix^{\Transpose} & \zeromatrix
		\end{pmatrix}.
	\end{align*}
	The  \emph{vertex Laplacian} of   $\Graph$ is the $(\abs{V_\Graph}-1)\times (\abs{V_\Graph}-1)$--matrix
	\begin{align*}
		\laplacian &\defas  {\incidencematrix}^{\Transpose} \edgematrix^{-1}  {\incidencematrix} =\laplacian^{\Transpose}.
	\end{align*}
\end{definition}

By construction, the determinants of the Laplacian and of the expanded Laplacian coincide up to a factor.
\begin{definition}\label{def:Symanzik_polynomial}
	The  \emph{Symanzik polynomial} of the graph $\Graph$ is
	\begin{align*}
		\firstsymanzik_\Graph\defas \det \explaplacian = \det \laplacian \cdot \prod_{e\in E_\Graph} a_e
	\end{align*}
\end{definition}

The following statement about determinants of minors of incidence matrices goes back to Kirchhoff \cite{kirchhoff_ueber_1847}. The cycle incidence version we attribute to Tutte \cite{tutte_lectures_1965}.
\begin{lemma}[\cite{kirchhoff_ueber_1847}, \cite{tutte_lectures_1965}~Theorem 5.46]\label{lem:matrix-tree}
  Let $U$ denote a subset of edges of $\Graph$ such that $\abs{U} = \abs{V_{\Graph}}-1$.
  Let $\incidencematrix[U]$ and $\cyclebasis[\overline{U}]$ denote the square submatrices as in \cref{submatrix_rows}, where the complement $\overline{U} \defas E_{\Graph} \setminus U$.
  Then
  \begin{align*}
    \det\incidencematrix[U] = \begin{cases} \pm 1 & \text{ if $U$ is a spanning tree} \\ 0 & \text{otherwise} \end{cases},
      \qquad
    \det\cycleincidencematrix[\overline{U}] = \begin{cases} \pm 1 & \text{ if $U$ is a spanning tree} \\ 0 & \text{otherwise} \end{cases}
  \end{align*}
\end{lemma}

The (weighted) matrix tree theorem is   the well-known consequence of \cref{lem:matrix-tree} (cf.~\cite{moon_counting_1970,maurer_matrix_1976})  that the Symanzik polynomial is given by a sum over all spanning trees of the graph,
\begin{align}\label{symanzik_polynomial_trees}
  \det\laplacian = \sum_{\substack{T \\ \text{spanning tree}}} \prod_{e \in T} \frac{1}{a_e}, \qquad \qquad
  \firstsymanzik_{\Graph} = \sum_{\substack{T \\ \text{spanning tree}}} \prod_{e \not\in T} a_e .
\end{align}

We can also consider another version of a Laplacian matrix, this time defined via a cycle incidence matrix.
\begin{definition}\label{def:cycle_laplacian}
	The \emph{cycle Laplacian}\footnote{This matrix is called \emph{dual Laplacian} in \cite{brown_unstable_2024}, whereas in \cite{brown_invariant_2021} it is  just \emph{Laplacian} and $\laplacian$ is called dual Laplacian. We use the names \enquote{cycle Laplacian} and \enquote{vertex Laplacian} to avoid any potential confusion.  } of a graph $\Graph$, with respect to a basis $\mathcal{C}=(C_1,\ldots,C_{\loopnumber})$ of the cycle space $\cyclespace(\Graph)$ (\cref{def:cycle_incidence_matrix}), is the symmetric $\loopnumber\times\loopnumber$ matrix
	\begin{equation*}
		\duallaplacian_\mathcal{C} \defas \cycleincidencematrix^\Transpose \edgematrix \cycleincidencematrix.
	\end{equation*}
\end{definition}

The cycle Laplacian has an analogue of \cref{symanzik_polynomial_trees}. As this statement is much less known in the literature, we give an explicit proof.
\begin{lemma}\label{lem:dual-det}
	For any choice of cycle incidence matrix $\cycleincidencematrix$ of the graph $\Graph$,
	\begin{align*}
    \det{\duallaplacian_{\cycleincidencematrix}} =\sum_{\substack{T \\ \textnormal{spanning tree}}} \prod_{e \not\in T} a_e = \firstsymanzik_{\Graph}
	\end{align*}
	coincides with the Symanzik polynomial (\cref{def:Symanzik_polynomial}.)
\end{lemma}
\begin{proof}
	Expanding the determinant of $\duallaplacian_{\cycleincidencematrix}$ via the Cauchy-Binet formula,
	\begin{align*}
		\det\duallaplacian_{\cycleincidencematrix}
		= \det\left(\cycleincidencematrix^{\Transpose}\edgematrix\cycleincidencematrix\right)
		&= \sum_{\substack{U \subseteq E_{\Graph}\\ \abs{U} = \loopnumber}} \det\left(\cycleincidencematrix^{\Transpose}[-,U] \right) \cdot \det\left(\left(\edgematrix\cycleincidencematrix\right)[U,-]\right) \\
		&= \sum_{\substack{U \subseteq E_{\Graph}\\ \abs{U} = \loopnumber}} \det\left(\cycleincidencematrix[U] \right) \cdot \det\left(\cycleincidencematrix[U]\right) \cdot \prod_{e \in U} a_e
		= \sum_{\substack{T \\ \text{spanning tree}}} \prod_{e \not\in T} a_e
		= \firstsymanzik_{\Graph}
	\end{align*}
  where the second to last equality follows from  \cref{lem:matrix-tree}, and the last equality is \cref{symanzik_polynomial_trees}.
\end{proof}

It might appear that the properties of the cycle Laplacian (\cref{def:cycle_laplacian}) could depend on the particular choice of cycle incidence matrix $\cycleincidencematrix$ used in its definition. Indeed, choosing a different $\cycleincidencematrix'$ amounts to choosing a different basis in cycle space, hence a transformation   $\mathcal{C}'=\mathcal{C}P$ for some invertible matrix $P\in \GL_\loopnumber(\Z)$. This implies a change of the corresponding cycle Laplacian by conjugation,
\begin{equation}\label{eq:lapmat-trans}
	\duallaplacian_{\mathcal{C}'} = P^\Transpose \duallaplacian_{\mathcal{C}} P.
\end{equation}
The determinant of the cycle Laplacian then changes according to
\begin{align*}
\det \left( \cyclelaplacian_{\cycleincidencematrix'} \right) = \det \left( P \right) ^2 \cdot \det \left( \cyclelaplacian_\cycleincidencematrix \right) .
\end{align*}
The fact that both $\cycleincidencematrix$ and $\cycleincidencematrix'$ have integer entries $\pm 1$, and that $P$ is invertible, means that   $\det(P)=\pm 1$. Consequently,  the determinant of the cycle Laplacian, which is the Symanzik polynomial (\cref{def:Symanzik_polynomial}), does not depend on the choice of $\cycleincidencematrix$.

As every non-empty connected graph has a spanning tree, when all the Schwinger parameters are strictly positive $a_e > 0$ by \cref{symanzik_polynomial_trees} we have that the Symanzik polynomial is non-zero and thus both Laplacian matrices are invertible. This is also easy to see from the fact that, by definition or as a corollary of \cref{lem:matrix-tree}, $\incidencematrix$ has rank $(\abs{V_\Graph}-1)$ and $\cycleincidencematrix$ has rank $\loopnumber$.

\begin{proposition}\label{lem:laplacian_cyclelaplacian_inverse}
	The inverses of the vertex Laplacian $\laplacian$ (\cref{def:Laplacian}) and the cycle Laplacian $\cyclelaplacian_\cycleincidencematrix$ (\cref{def:cycle_laplacian}), regardless of the choice of $\cycleincidencematrix$, are related via
	\begin{align*}
    \edgematrix -   \incidencematrix \laplacian^{-1} \incidencematrix^\Transpose = \edgematrix \cycleincidencematrix \cyclelaplacian^{-1}_{\cyclebasis} \cycleincidencematrix^\Transpose \edgematrix.
	\end{align*}
\end{proposition}
\begin{proof}
	Consider the following two matrix products:
	\begin{align*}
		P_\cutspace &\defas   \edgematrix^{-1} \incidencematrix \laplacian^{-1} \incidencematrix^\Transpose,\qquad
    P_\cyclespace   \defas  \cycleincidencematrix \cyclelaplacian_{\cyclebasis}^{-1} \cycleincidencematrix^\Transpose \edgematrix.
	\end{align*}
	$P_F$ has rank $(\abs{V_\Graph}-1)$ and $P_\cyclespace$ has rank $\loopnumber$.
	From \cref{def:Laplacian}, one finds
	\begin{align*}
	P_\cutspace ^2 &=  \edgematrix^{-1} \incidencematrix \laplacian^{-1}~ \incidencematrix^\Transpose \edgematrix^{-1} \incidencematrix ~\laplacian^{-1} \incidencematrix^\Transpose =  \edgematrix^{-1} \incidencematrix \laplacian^{-1} \incidencematrix^\Transpose = P_\cutspace.
	\end{align*}
	Likewise, it follows from \cref{def:cycle_laplacian} and \cref{lem:pathmatrix_incidencematrix_orthogonal} that $P_\cyclespace^2=P_\cyclespace$ and $P_\cutspace P_\cyclespace = \zeromatrix = P_\cyclespace P_\cutspace$. This certifies that $P_\cutspace$ and $P_\cyclespace$ are orthogonal projection operators onto subspaces of edge space. Moreover, by \cref{def:loopnumber}, the sum of the two projectors has rank $\abs{E_\Graph}$, and therefore $	P_\cutspace + P_\cyclespace = \identitymatrix_{\abs{E_\Graph}}$.
	Multiplying from the left by $\edgematrix^{-1}$ yields the claimed equation.
\end{proof}

\subsection{Concatenated matrices and signs} \label{sec:concatenated_matrices}

We now consider square matrices which arise from concatenating pairs of the   matrices defined in \cref{sec:incidence_path_matrices}. These combined matrices have size $\abs{E_\Graph}\times \abs{E_\Graph}$, and they have full rank.

\begin{lemma}\label{lem:determinant_concatenated_1}
	Let $\incidencematrix$ be the vertex incidence matrix and $\cycleincidencematrix$ the cycle incidence matrix. Then
	\begin{enumerate}
		\item   { \hfil 	$ \displaystyle
			\det  \concatm{\cycleincidencematrix}{\edgematrix^{-1}\incidencematrix}   ^2 =   \firstsymanzik^2_\Graph \prod_e a_e^{-2}$.  }
		\item Let $R$ be any matrix such that $R^\Transpose \incidencematrix= \identitymatrix_{\abs{V_\Graph} -1 }$. Then
		\begin{align*}
			\det   \concatm{\cycleincidencematrix}{R}   \det   \concatm{\cycleincidencematrix}{\edgematrix^{-1} \incidencematrix}   &=   \psi_\Graph \prod_e a_e^{-1} .
		\end{align*}
    In particular, by (1), $\det\concatm{\cycleincidencematrix}{R} = \pm 1$.
		\item Let $S$ be any matrix such that $S^\Transpose \edgematrix \cycleincidencematrix= \identitymatrix_{\loopnumber }$. Then
		\begin{align*}
			\det  \concatm{S}{\edgematrix^{-1} \incidencematrix}   \det   \concatm{\cycleincidencematrix}{\edgematrix^{-1} \incidencematrix}   &=    \psi_\Graph \prod_e a_e^{-2}.
		\end{align*}
    In particular, by (1), $\det\concatm{S}{\edgematrix^{-1}\incidencematrix} = \pm \prod_e a_e^{-1}$.
	\end{enumerate}
\end{lemma}
\begin{proof}
	By the orthogonality of $\cycleincidencematrix$ and $\incidencematrix$ (\cref{lem:pathmatrix_incidencematrix_orthogonal}), the following matrix product is block diagonal:
	\begin{align*}
  \concatm{\cycleincidencematrix}{\edgematrix^{-1}\incidencematrix}^\Transpose \edgematrix \concatm{\cycleincidencematrix}{\edgematrix^{-1}\incidencematrix}
  &= \begin{pmatrix}
      \cycleincidencematrix^\Transpose \edgematrix  \cycleincidencematrix & \zeromatrix \\
      \zeromatrix & \incidencematrix^\Transpose 	\edgematrix^{-1}\incidencematrix
		\end{pmatrix}=\begin{pmatrix}
      \cyclelaplacian_{\cyclebasis} & \zeromatrix \\
      \zeromatrix & \laplacian
		\end{pmatrix}.
	\end{align*}
	Likewise
	\begin{align*}
		\concatm{\cycleincidencematrix}{R}^\Transpose \edgematrix \concatm{\cycleincidencematrix}{\edgematrix^{-1} \incidencematrix}&= \begin{pmatrix}
			\cycleincidencematrix^\Transpose \edgematrix \cycleincidencematrix & \cycleincidencematrix^\Transpose \incidencematrix\\
			R^\Transpose \edgematrix \cycleincidencematrix & R^\Transpose \edgematrix \edgematrix^{-1} \incidencematrix
		\end{pmatrix}=\begin{pmatrix}
      \cyclelaplacian_{\cyclebasis} & \zeromatrix \\
			R^\Transpose \edgematrix \cycleincidencematrix & \identitymatrix
		\end{pmatrix}\\
		\concatm{S}{\edgematrix^{-1} \incidencematrix}^\Transpose \edgematrix \concatm{\cycleincidencematrix}{\edgematrix^{-1} \incidencematrix}&= \begin{pmatrix}
			S^\Transpose \edgematrix \cycleincidencematrix & S^\Transpose \incidencematrix\\
			\incidencematrix^\Transpose \edgematrix^{-1} \edgematrix \cycleincidencematrix &  \incidencematrix^\Transpose \edgematrix^{-1} \edgematrix \edgematrix^{-1} \incidencematrix
		\end{pmatrix}=\begin{pmatrix}
			\identitymatrix & S^\Transpose \incidencematrix \\
			\zeromatrix & \laplacian
		\end{pmatrix}.
	\end{align*}
  On the left hand side, taking the determinant gives the product of the determinants of the three block matrices. Since the determinant is invariant under transposition, the claim follows. Identify the Symanzik polynomials according to $\det (\laplacian)  = \firstsymanzik_\Graph \prod a_e^{-1}$ (\cref{def:Symanzik_polynomial}) and $\det (\cyclelaplacian_{\cyclebasis}) = \firstsymanzik_\Graph$ (\cref{lem:dual-det}).
\end{proof}

The statements of \cref{lem:determinant_concatenated_1} have useful corollaries when they are evaluated at $\edgematrix=\identitymatrix$, that is, when $a_e=1$ for all $e$. In that case, by \cref{symanzik_polynomial_trees} $\firstsymanzik_\Graph$ is the number of spanning trees of $\Graph$, which, in particular, is a positive integer. We will need the following:
\begin{corollary}\label{lem:determinant_spanning_trees}
  For any choice of $\abs{E_G} \times (\abs{V_G}-1)$ matrix $R$ such that $R^{\Transpose} \incidencematrix = \identitymatrix_{\abs{V_G}-1}$,
	\begin{align*}
    \det \concatm{\cycleincidencematrix}{\incidencematrix}
    &= \det \concatm{\cycleincidencematrix}{R}  \times \left(  \textnormal{\# of spanning trees in $\Graph$}\right),
	\end{align*}
  where $\det\concatm{\cycleincidencematrix}{R} = \pm 1 $.
\end{corollary}

Note that such matrix $R$ always exists, since by \cref{lem:pathmatrix_incidencematrix_orthogonal} the path matrix $\pathmatrix$ (\cref{def:pathmatrix}) satisfies $\pathmatrix^\Transpose \incidencematrix = \identitymatrix_{\abs{V_\Graph}-1}$.
\Cref{lem:determinant_spanning_trees} is nice on its own, but most importantly, it allows us to fix an alternating sign in the main proof and relates the minors of $\incidencematrix$ and $\cycleincidencematrix$.

\begin{lemma}\label{lem:signs}
  Let $\Graph$ be a connected graph where all labellings are arbitrary but fixed. Let $\cycleincidencematrix$ be any choice of a cycle incidence matrix and let $\pathmatrix$ be any path matrix. Then, for all sets $T\subseteq E_\Graph$ of $\abs{V_\Graph}-1$ edges with $\overline{T} \defas E_{\Graph} \setminus T$
  \begin{align}\label{eq:inc-cycinc-relation}
  \det \left( \incidencematrix[T] \right) &= (-1)^{\frac{\loopnumber(\loopnumber+1)}{2}} \det\concatm{\cycleincidencematrix}{\pathmatrix} \cdot  (-1)^{\sum_{e\notin T} e} \det \left( \cycleincidencematrix[\overline T] \right) ,
	\end{align}
  where the sign $(-1)^{\frac{\loopnumber(\loopnumber+1)}{2}} \det\concatm{\cycleincidencematrix}{\pathmatrix}$ is independent of the choice of $T$. In particular, both sides of the equation are zero unless $T$ is a spanning tree in $\Graph$.
\end{lemma}
\begin{proof}
  First, for any subset $T \subseteq E_{\Graph}$ of $\abs{V_{\Graph}}-1$ edges that contains a cycle, we have that by   \cref{lem:matrix-tree}, $\det\left(\incidencematrix[T]\right) = \det\left(\cycleincidencematrix[\overline{T}]\right) = 0$.
  As $\det\concatm{\cycleincidencematrix}{\pathmatrix}$ can never be $0$ by \cref{lem:determinant_spanning_trees}, \longequation \cref{eq:inc-cycinc-relation} \shortequation holds trivially for such $T$ and thus both sides are non-zero if and only if $T$ is a spanning tree.

  Fix $T \subseteq E_{\Graph}$ to be some spanning tree of $\Graph$.
  As $\det\concatm{\cycleincidencematrix}{\pathmatrix}$ is invariant under the choice of $\pathmatrix$, let $\pathmatrix$ be the path matrix such that column $P_i$ corresponds to the unique directed path from $v_{\star}$ to $v_i$ contained in $T$.
  In particular, $\pathmatrix[\overline{T}] = 0$ and thus $\identitymatrix_{\abs{V_\Graph}-1} =\pathmatrix^{\Transpose}\incidencematrix = \left(\pathmatrix[T]\right)^{\Transpose}\incidencematrix[T]$, immediately obtaining that $\det\left(\pathmatrix[T]\right) = \det\left(\incidencematrix[T]\right)$.

  Thus we can rearrange the rows of $\concatm{\cycleincidencematrix}{\pathmatrix}$, which correspond to the edges $E_{\Graph}$, to compute
  \begin{align*}
    \det\concatm{\cycleincidencematrix}{\pathmatrix}
    = (-1)^{\sum_{i=1}^{\loopnumber} f_i - i} \det\begin{pmatrix}
      \cycleincidencematrix[\overline{T}] & \zeromatrix \\
      \cycleincidencematrix[T] & \pathmatrix[T]
    \end{pmatrix}
    &= (-1)^{\sum_{i = 1}^{\loopnumber} f_i - i} \det\left(\cycleincidencematrix[\overline{T}]\right) \cdot \det\left(\pathmatrix[T]\right) \\
    &= (-1)^{\sum_{e \not\in T} e} (-1)^{\frac{\loopnumber(\loopnumber+1)}{2}}  \det\left(\cycleincidencematrix[\overline{T}]\right) \cdot \det\left(\incidencematrix[T]\right),
  \end{align*}
  where $\overline{T} = \{f_1 < \cdots < f_{\loopnumber} \}$ are taken in increasing order and the sign in the first equality comes from rearranging the rows such that the edges in $\overline{T}$ come before the edges in $T$.
\end{proof}

Note that the sign $\det\concatm{\cycleincidencematrix}{\pathmatrix}$ is not unexpected and can be interpreted as exactly the sign needed to relate the two notions of orientation data on the graph $\Graph$ required to define the vertex and cycle incidence matrices.
Namely, given a fixed ordering of the edges $E_{\Graph}$, the signs tells us how to translate from an ordering of the vertices $V_{\Graph}$ and fixing directions of edges, which is needed to define $\incidencematrix$, to an ordering of a given cycle basis, which is needed to define $\cyclebasis$.
See~\cite[\S4.4]{brown_unstable_2024} for further discussion; there, the matrix denoted $A$ amounts to our matrix $\concatm{\cycleincidencematrix}{\pathmatrix}$, but with a different choice of special vertex ($v_{\star} = v_1$).

\medskip

Finally, we give the proof for the change of sign claimed in \cref{lem:alpha_2-valent}.
\begin{lemma}\label{lem:subdivision_sign}
  Let $G'$ be the graph obtained from $G$ by subdividing the edge $e$ into edges $e'$ and $e''$.
  We take the new edges $e', e''$ in the same direction as $e$, let $e'$ inherit the label of $e$ and let $e''$ have the label $e+1$, shifting all subsequent edge labels by one unit.
  Let the newly added vertex be labelled $v_1$, where all existing vertices are shifted by one unit.
  Then, the sign of the topological form changes by $(-1)^{e+1}$.
\end{lemma}
\begin{proof}
  All we need to show is that we have the relation
		$\det \concatm{\cycleincidencematrix'}{\pathmatrix'}= (-1)^{e+1} \det \concatm{\cycleincidencematrix}{\pathmatrix}$ between the matrices for $G'$ and $G$.
  Then by \cref{thm:main_thm} and \cref{lem:2-valent}, we immediately obtain
  \begin{align*}
    \alpha_{\Graph'}
    = \frac{\det{\concatm{\cyclebasis'}{\pathmatrix'}}}{2^{\loopnumber}} \phi_{G'}
    = (-1)^{e+1}\frac{\det{\concatm{\cyclebasis}{\pathmatrix}}}{2^{\loopnumber}} s_e^*(\phi_{G})
    = (-1)^{e+1} s_e^*(\alpha_{\Graph}).
  \end{align*}

  Suppose the edge $e = v \to w$ is directed from vertex $v$ to vertex $w$.
  The subdivided edges in $G'$ are thus directed $e' = v \to v_1$ and $e'' = v_1 \to w$.
  Recall that in the Pfaffian case (\cref{lem:2-valent}), the cycle basis defining $\cyclebasis'$ is induced from $\cyclebasis$ by replacing the edge $e$ with the path $e' e''$.
  As we can choose any path matrix $\pathmatrix'$ for $\Graph'$, similarly, we can take $\pathmatrix'$ to be induced from $\pathmatrix$ by replacing the edge $e$ with the path $e' e''$ in any path that contains $e$.
  We further need to specify a path from $v_{\star}$ to the newly created vertex $v_1$.
  Here we can take the path $v_{\star} \to v$ used in $\pathmatrix$ and add the edge $+e'$.

  Translating into matrices, $\concatm{\cyclebasis'}{\pathmatrix'}$ is obtained from $\concatm{\cyclebasis}{\pathmatrix}$ by adding a new row $e''$ in position $e+1$ and a new column $v_1$ in position $\loopnumber+1$.
  Row $e''$ coincides with row $e$ (which corresponds to edge $e'$) everywhere except in column $v_1$ of $\pathmatrix'$, where row $e$ has a $+1$ while row $e''$ has a $0$.
  Similarly, the columns $v_1$ and $v$ coincide everywhere except in row $e$.
  Thus, expanding the determinant of $\concatm{\cyclebasis'}{\pathmatrix'}$ along column $v_1$ leaves only one non-zero contribution, which is exactly $(-1)^{e+\loopnumber+1} \det\concatm{\cyclebasis}{\pathmatrix}$, as all other terms involve matrices with two coinciding rows.
  Therefore, as $\loopnumber$ is even, we find that
		$\det \concatm{\cycleincidencematrix'}{\pathmatrix'} = (-1)^{e+1} \det \concatm{\cycleincidencematrix}{\pathmatrix}$.
\end{proof}

\subsection{Dodgson polynomials}

Consider sets $A $ and $B$ of integers between 1 and $(\abs{E_\Graph}+\abs{V_\Graph}-1)$, not necessarily disjoint. These sets correspond to sets of rows and columns of $\explaplacian$ (\cref{def:Laplacian}), and using the notation of \cref{def:mat-minors}, we are interested in the following minors of $\explaplacian$ as defined in \cite{brown_periods_2010}.

\begin{definition}\label{def:Dodgson_polynomial}
	Let $\abs{A}=\abs{B}$. The \emph{Dodgson polynomial} of a graph $\Graph$ is the determinant of the submatrix of $\explaplacian$ (\cref{def:Laplacian}) where rows $A$ and columns $B$ have been removed,
	\begin{align*}
		\psi^{A,B}_\Graph &\defas \det \explaplacian(A,B).
	\end{align*}
\end{definition}

A closely related set of polynomials can be defined in terms of spanning forests \cite{brown_spanning_2011,golz_dodgson_2019}. Dodgson/spanning forest polynomials have been used e.g. in the study of period invariants \cite{brown_k3_2012,schnetz_geometries_2021} and for the manipulation of of parametric Feynman integrands \cite{bogner_feynman_2010,bellon_numerators_2024,golz_new_2017,schnetz_geometries_2021}.
These two types of graph polynomials are related via the following expansion for Dodgson polynomials, which has appeared up to sign in \cite{brown_periods_2010,brown_spanning_2011}, and can further be expressed as a signed sum of spanning forest polynomials.

\begin{lemma}\label{lem:dodgson_expansion}
  Let $A, B \subseteq E_{\Graph}$ be sets of edges with $\abs{A} = \abs{B}$.
  Let $\mathcal{U}$ be the set of subsets $U$ of edges in $E_\Graph \setminus (A \cup B)$ such that $\abs{U} = \loopnumber - \abs{A}$.
  Then the Dodgson polynomial is an alternating sum over edges not in spanning trees,
  \begin{align}\label{eq:dodg_exp}
    \dodgson^{A,B}_{\Graph} = \sum_{U \in \mathcal{U}} (-1)^{\sum_{e \in U} \abs{A_{<e}} - \abs{B_{<e}}} \det\big( \incidencematrix(U \cup A,\emptyset) \big) \det\big(\incidencematrix(U \cup B, \emptyset)\big) \prod_{e \in U} a_e,
  \end{align}
  where $A_{<e} = \{a \in A | a < e\}$ and similarly for $B_{<e}$.
  In particular, the only terms that contribute are those such that both $U \cup A$ and $U \cup B$ are spanning tree complements in $\Graph$.
  In the special case $A = B = \emptyset$, we recover \cref{symanzik_polynomial_trees}.
\end{lemma}
\begin{proof} 
  Expanding the determinant $\dodgson^{A,B}_{\Graph}$ via the Leibniz formula, and based on the structure of $\explaplacian(A,B)$, the only non-zero terms come from permutations that factors the determinant as such:
  \begin{align*}
    &\dodgson^{A,B}_{\Graph}
    = \sum_{U \subseteq E_{\Graph} \setminus (A \cup B)} \det\left(\edgematrix(A,B)[U,U]\right)  \det\begin{pmatrix}\zeromatrix & \incidencematrix(U \cup A, \emptyset) \\ -\incidencematrix^{\Transpose}(\emptyset,U \cup B) & \zeromatrix \end{pmatrix} \\
    &=  \sum_{U \subseteq E_{\Graph} \setminus (A \cup B)} \det(\sigma_U) \left(\prod_{e \in U} a_e \right) (-1)^{\abs{V_{\Graph}}-1+(\abs{V_{\Graph}}-1)(\abs{E_{\Graph}} - \abs{U}-\abs{A})}\det\begin{pmatrix}\incidencematrix^{\Transpose}(\emptyset, U \cup B) & \zeromatrix \\ \zeromatrix & \incidencematrix(U \cup A,\emptyset) \end{pmatrix} \\
    &= \sum_{U \in \mathcal{U}} \det(\sigma_U) \left(\prod_{e \in U} a_e \right)  \det\left(\incidencematrix(U \cup A, \emptyset)\right) \det\left(\incidencematrix(U \cup B, \emptyset)\right).
  \end{align*}
  Here, since $\edgematrix(A,B)[U,U]$ may no longer be diagonal, its determinant induces a sign $\det(\sigma_U)$ for some permutation matrix $\sigma_U$.
  For the last equality, $\incidencematrix(U \cup A, \emptyset)$ is an $(\abs{E_{\Graph}} - \abs{U} - \abs{A}) \times (\abs{V_{\Graph}}-1)$ matrix, which by Euler's formula (\cref{def:loopnumber}) is square and has non-zero determinant when $\abs{U} = \loopnumber - \abs{A}$ (recall $\incidencematrix$ is of rank $\abs{V_{\Graph}}-1$), similarly for $\incidencematrix^{\Transpose}(\emptyset, U \cup B)$.
  In particular when $\abs{U} = \loopnumber - \abs{A}$,  \cref{lem:matrix-tree} tells us that the product $\det\left(\incidencematrix(U \cup A, \emptyset) \right) \det\left(\incidencematrix(U \cup B, \emptyset)\right)$ is non-zero, and furthermore equal to $\pm 1$, if and only if both $U \cup A$ and $U \cup B$ are spanning tree complements.

  Finally to determine $\det(\sigma_U)$, we just need to count the number of transpositions it takes to bring $a_e$ to the diagonal of $\edgematrix(A,B)[U,U]$ for $e \in U$.
  This count is exactly the difference between the number of rows removed from $\edgematrix$ that is above row $e$ and the number of columns removed from $\edgematrix$ that is before column $e$, which is given by $\abs{\{a \in A | a < e\}} - \abs{\{b \in B | b < e \}}$.
\end{proof}

Recall that the matrix $\explaplacian$  (\cref{def:Laplacian}) has edge- and vertex indices, where the top left block is indexed by edges.
The cycle Laplacian $\duallaplacian_{\cyclebasis}$ (\cref{def:cycle_laplacian}) is indexed by cycles, and the mapping between cycles and edges is done by $\cycleincidencematrix$. Analogously, the vertex Laplacian $\laplacian$ (\cref{def:Laplacian}) is indexed by vertices, and the incidence matrix $\incidencematrix$ maps between edges and vertices.

\begin{proposition}\label{lem:laplacian_inverse_dodgson}
	The entries of the inverse vertex Laplacian and of the inverse cycle Laplacian are given by Dodgson polynomials,
	\begin{align*}
    \left( \laplacian^{-1} \right) _{v_i,v_j}
      &= \frac{(-1)^{i+j}\dodgson_{\Graph}^{v_i,v_j}}{\firstsymanzik_{\Graph}},
    \qquad\qquad\left( \cycleincidencematrix \cyclelaplacian_{\cyclebasis}^{-1} \cycleincidencematrix^\Transpose \right) _{e_i, e_j}
      = \frac{(-1)^{i+j}\dodgson_{\Graph}^{e_i,e_j}}{\firstsymanzik_{\Graph}}, \\
    \left( \edgematrix^{-1}	\incidencematrix \laplacian^{-1}\incidencematrix^\Transpose  \right) _{e_i,e_j}
      &= \delta_{e_i,e_j} - \frac{(-1)^{i+j}   a_{e_j} \dodgson_{\Graph}^{e_i,e_j}}{\firstsymanzik_{\Graph}} .
	\end{align*}
  In the last equation, the diagonal entry $e_i=e_j$ can equivalently be written as $\frac{ \firstsymanzik_{\Graph / {e_i}}}{\firstsymanzik_{\Graph}} $.
\end{proposition}
\begin{proof}
	The Dodgson polynomials with single-element indices are just the cofactors of $\explaplacian$, therefore, they are entries of the inverse,
	\begin{align*}
    \left( \explaplacian^{-1} \right) _{i,j} = (-1)^{i+j} \frac{\dodgson_{\Graph}^{i,j}}{\firstsymanzik_{\Graph}}
	\end{align*}
	At the same time, thanks to the block form of $\explaplacian$, one can compute the inverse in block form:
	\begin{align*}
		\explaplacian = \begin{pmatrix}
			\edgematrix & \incidencematrix \\
			-\incidencematrix^{\Transpose} &0
		\end{pmatrix}\quad \Rightarrow \quad
		\explaplacian^{-1} &= \begin{pmatrix}
			\edgematrix^{-1} + \edgematrix^{-1}\incidencematrix \left( \incidencematrix^\Transpose \edgematrix^{-1} \incidencematrix \right) ^{-1}(-\incidencematrix^\Transpose)\edgematrix^{-1} & \quad -\edgematrix^{-1} \incidencematrix \left( \incidencematrix^\Transpose \edgematrix^{-1} \incidencematrix \right) ^{-1}\\
			-\left( \incidencematrix^\Transpose \edgematrix^{-1} \incidencematrix \right) ^{-1}(-\incidencematrix^\Transpose)\edgematrix^{-1} & \left( \incidencematrix^\Transpose \edgematrix^{-1} \incidencematrix \right) ^{-1}
		\end{pmatrix}\\
		&= \begin{pmatrix}
			\edgematrix^{-1} - \edgematrix^{-1}\incidencematrix \laplacian  ^{-1} \incidencematrix^\Transpose  \edgematrix^{-1} & \quad -\edgematrix^{-1} \incidencematrix \laplacian  ^{-1}\\
			\laplacian ^{-1} \incidencematrix^\Transpose \edgematrix^{-1} &\laplacian ^{-1}
		\end{pmatrix}\\
		&= \begin{pmatrix}
      \cycleincidencematrix \duallaplacian_{\cyclebasis}^{-1} \cycleincidencematrix^\Transpose & \quad -\edgematrix^{-1} \incidencematrix \laplacian  ^{-1}\\
			\laplacian ^{-1} \incidencematrix^\Transpose \edgematrix^{-1} &\laplacian ^{-1}
		\end{pmatrix}.
	\end{align*}
	In the last step, we have used \cref{lem:laplacian_cyclelaplacian_inverse}. Taking coefficients gives the first two equations. The third equation follows from taking coefficients of the top left block $\edgematrix^{-1} - \edgematrix^{-1}\incidencematrix \laplacian^{-1} \incidencematrix^\Transpose  \edgematrix^{-1}$
  in the intermediate form of $\explaplacian^{-1}$, and multiplying by $\edgematrix$.

	To prove the alternative formula for the diagonal entry, recall that setting $a_e=0$ in $\firstsymanzik_\Graph$ means to contract the edge $e$ in the graph. When $e=s \rightarrow t$ is any edge, a determinant expansion of $\explaplacian$ along row $e$ yields
	\begin{align*}
		\firstsymanzik_{\Graph/e}=\firstsymanzik_\Graph\rvert_{a_e=0} = \det \explaplacian \rvert_{a_e=0} = (-1)^{t+e+\abs{E}} a_e \dodgson^{e,t}_\Graph - (-1)^{s+e+\abs{E}} a_e \dodgson^{e,s}_\Graph.
	\end{align*}
	Consequently, for the case  $f = e$, we have that
	\begin{align}\label{ILI_proof_eq1}
		\left( \incidencematrix\laplacian^{-1}\incidencematrix^{\Transpose} \right) _{e,e}
    &=  \frac{(-1)^{t+e+\abs{E}} a_e \dodgson_{\Graph}^{e,t} - (-1)^{s+e+\abs{E}} a_e \dodgson_{\Graph}^{e,s}}{\firstsymanzik_\Graph}
		= \frac{a_e \firstsymanzik_{\Graph/e}}{\firstsymanzik_\Graph}.
	\end{align}
	The diagonal Dodgson edge polynomial amounts to deletion of an edge, $\dodgson^{e,e}_{\Graph}=\firstsymanzik_{\Graph \setminus e}$, and the Symanzik polynomial satisfies the contraction-deletion identity
	$\firstsymanzik_\Graph =a_e \firstsymanzik_{\Graph \setminus e}  + \firstsymanzik_{\Graph / e}$.
	Consequently, the last term of \cref{ILI_proof_eq1} can be rewritten as
	\begin{align*}
		\frac{a_e \firstsymanzik_{\Graph /e}}{ \firstsymanzik_\Graph}= a_e- \frac{a_e^2 \firstsymanzik_{\Graph \setminus e} } { \firstsymanzik_\Graph}= a_e- \frac{a_e^2 \dodgson^{e,e}_\Graph } { \firstsymanzik_\Graph}.
	\end{align*}
	Noticing that $(-1)^{e+e}=(-1)^{2e}=1$, this has the same form as the off-diagonal form, up to a Kronecker delta. Finally, the factor $a_e$ in both expressions gets cancelled when we multiply by $\edgematrix^{-1}$ from the left.
\end{proof}

\Cref{lem:laplacian_inverse_dodgson} and its proof are also of interest as we obtain several (new) linear relations between edge and vertex Dodgson polynomials as a consequence.
By taking other coefficients of $\explaplacian^{-1}$ analogously to the proof of \cref{lem:laplacian_inverse_dodgson}, for example, results in an alternative proof of \cite[Lemmas~11, 12]{balduf_combinatorial_2024}.
The last equation in~\cref{lem:laplacian_inverse_dodgson} implies the following new identities:
\begin{corollary} ~\\[-.5cm]
	\begin{enumerate}
		\item For any edge $e = s \to t$ where $s$ and $t$ are distinct vertices,
		\[ a_e \dodgson_{G/e} = \dodgson_{\Graph}^{s,s} + \dodgson_{\Graph}^{t,t} - 2(-1)^{s+t}\dodgson_{\Graph}^{s,t}. \]
		\item   By the orthogonality of $\cyclebasis$ and $\incidencematrix$, for any cycle $C$ identified with its edge incidence vector $(c_j)$ and for any edge $e_i$,
		\begin{align*}
			c_i \dodgson_{\Graph} - \sum_{e_j \in E_{\Graph}} (-1)^{e_i+e_j} c_j a_{e_j} \dodgson_{\Graph}^{e_i, e_j}=0.
		\end{align*}
	\end{enumerate}

\end{corollary}

Finally, for the special case where $\cycleincidencematrix$ is defined in terms of a fundamental cycle basis, we can directly express the entries of $\cyclelaplacian_{\cyclebasis}^{-1}$ in terms of edge Dodgson polynomials.

\begin{corollary}\label{lem:dual-inverse}
  Let $\fundamentalcbasis$ be the cycle incidence matrix corresponding to a fundamental cycle basis $\{C_1, \ldots, C_{\loopnumber}\}$ defined via a spanning tree $T$ (see \cref{def:fundamentalcyclebasis}), where each cycle $C_i$ corresponds to an edge $f_i \notin T$.
  Then the $(i,j)$-entry of $\cyclelaplacian_{\cyclebasis}^{-1}$ is given by edge Dodgson polynomials (which, by \cref{def:Dodgson_polynomial}, are independent of the choice of cycle basis),
  \begin{align*}
    \left( \duallaplacian_{\fundamentalcbasis}^{-1} \right) _{C_i,C_j} =  (-1)^{f_i+f_j} \frac{\dodgson_{\Graph}^{f_i, f_j}}{\firstsymanzik_{\Graph}},
  \end{align*}
\end{corollary}
\begin{proof}
  $f_i$ being the defining edge of the cycle $C_i$  implies that $C_i$ is the only cycle that contains $f_i$, and the directions of $C_i$ and $f_i$ coincide. Therefore, the entry $\fundamentalcbasis_{f_i, C_i}=+1$ is the only non-zero entry in the row $f_i$ of $\fundamentalcbasis$. Phrased differently, if $f_i$ is one of the defining edges, then the vector $\fundamentalcbasis^\Transpose[-, f_i]$ is a unit basis vector for $C_i$. Now  	consider the equation
	\begin{align*}
    \frac{(-1)^{i+j}\dodgson_{\Graph}^{e_i,e_j}}{\firstsymanzik_{\Graph}}=\left( \cycleincidencematrix \cyclelaplacian_{\cyclebasis}^{-1} \cycleincidencematrix^\Transpose \right) _{e_i, e_j}
	\end{align*}
	from \cref{lem:laplacian_inverse_dodgson}. It holds for arbitrary edges. However, if $\cycleincidencematrix$ is defined in terms of a fundamental cycle basis, and $f_i, f_j$ are defining edges of fundamental cycles, then
	\begin{align*}
    \left( \fundamentalcbasis \cyclelaplacian_{\fundamentalcbasis}^{-1} \fundamentalcbasis^\Transpose \right) _{f_i, f_j} &= \left(\cyclelaplacian_{\fundamentalcbasis}^{-1}\right)_{C_i,C_j}.
	\end{align*}
\end{proof}

\subsection{Dodgson representation of the Pfaffian form}

Recall from \cref{def:pff-form} that the Pfaffian form $\phi_{G}$ for the graph $\Graph$ is defined as
  \begin{align}\label{eq:pff-recall}
    \phi_{\Graph} = \frac{1}{(-2\pi)^{\loopnumber/2}} \frac{\Pf\left(\d\cyclelaplacian_{\cyclebasis} \cdot \cyclelaplacian_{\cyclebasis}^{-1} \cdot \d\cyclelaplacian_{\cyclebasis}\right)}{\sqrt{\dodgson_{\Graph}}},
  \end{align}
  where $\cyclelaplacian_{\cyclebasis}$ is the cycle Laplacian (\cref{def:cycle_laplacian}) corresponding to an ordered cycle basis $\cyclebasis$ (and its cycle incidence matrix).
  $\phi_\Graph$ is only non-zero when $\Graph$ has even loop number $\loopnumber$.

To express the Pfaffian form in terms of Dodgson polynomials, we will need the following minor summation formula of Pfaffians, which is an analogue of the Cauchy-Binet formula for determinants.
\begin{lemma}[\cite{ishikawa_minor_1995}~Theorem 1]\label{lem:minor-summ-pf}
  Let $A$ be an $m \times n$--matrix and $B$ be an $n \times n$ skew-symmetric matrix such that the entries of $B$ commute with the entries of $A$.
  Suppose $m \leq n$ and $m$ is even.
  Then
  \begin{align*}
    \Pf\big(A B A^{\Transpose} \big) = \sum_{U \in \binom{[n]}{m}} \det\big( A[-,U] \big) \Pf\big( B[U,U] \big),
  \end{align*}
  where the sum is over all subsets of size $m$ of $[n] = \{1, \cdots, n\}$.
\end{lemma}

\begin{theorem}\label{lem:phi_dodgson}
  For a connected graph $\Graph$ with even loop number $\loopnumber$ and any ordered cycle basis $\mathcal{C}$, $\phi_{\Graph}$ is a sum over spanning trees,
  \begin{align*}
    \scalemath{.95}{\phi_{\Graph}
    = \frac{1}{(-\pi)^{\frac{\loopnumber}{2}} 2^{\loopnumber} \left(\frac{\loopnumber}{2}\right)! \cdot \dodgson_{\Graph}^{\frac{\loopnumber+1}{2}}}\sum_{\substack{T \\ \textnormal{spanning tree}}} (-1)^{\sum_{e \notin T} e} \det\left(\cycleincidencematrix[\overline{T}]\right)  \left(\sum_{\sigma \in \mathfrak{S}_{\overline{T}}} \dodgson_{\Graph}^{\sigma(f_1),\sigma(f_2)} \cdots \dodgson_{\Graph}^{\sigma(f_{\loopnumber-1}), \sigma(f_\loopnumber)} \right) \bigwedge_{e \notin T} \d a_e,}
  \end{align*}
  where for each tree $T$, the edges $\{f_1,\ldots,f_{\loopnumber}\} \defas \overline{T} = E_{\Graph} \setminus T$ and the $\d a_e$ are taken in the increasing order corresponding to the edge orderings.
\end{theorem}

\begin{proof}
  Notice that we can rewrite
  \begin{align*}
    \d\cyclelaplacian_{\cyclebasis} \cdot \cyclelaplacian_{\cyclebasis}^{-1} \cdot \d\cyclelaplacian_{\cyclebasis}
    = \cyclebasis^{\Transpose} \left( \d\edgematrix \cdot \cyclebasis \cyclelaplacian_{\cyclebasis}^{-1} \cyclebasis^{\Transpose} \cdot \d\edgematrix \right) \cyclebasis.
  \end{align*}
  Letting $X \defas \d\edgematrix \cdot \cyclebasis \cyclelaplacian_{\cyclebasis}^{-1} \cyclebasis^{\Transpose} \cdot \d\edgematrix$,
  the minor summation formula (\cref{lem:minor-summ-pf}) and the matrix tree theorem (\cref{lem:matrix-tree}) immediately expands the Pfaffian of this matrix as a sum over spanning trees of $\Graph$,
  \begin{align}\label{eq:pff-minor-expn}
    \Pf\left(\d\duallaplacian_{\cyclebasis} \cdot \duallaplacian_{\cyclebasis}^{-1} \cdot \d\duallaplacian_{\cyclebasis} \right)
      &= \sum_{U \in \binom{E_{\Graph}}{\loopnumber}} \det\left( \cyclebasis^{\Transpose}[-,U] \right) \Pf\left( X[U,U] \right) \nonumber \\
      &= \sum_{\substack{T \\ \text{spanning tree}}} \det\left( \cyclebasis[\overline{T}] \right) \Pf\left( X[\overline{T},\overline{T}] \right),
  \end{align}
  where the first sum runs over all subsets $U \subseteq E_{\Graph}$ such that $\abs{U} = \loopnumber$ and $\overline{T} = E_{\Graph} \setminus T$.

  Using~\cref{lem:laplacian_inverse_dodgson}, we can identify the entries of the $\abs{E_{\Graph}} \times \abs{E_{\Graph}}$--matrix $X$ as
  \begin{align*}
    X_{i,j} = \left[ \d\edgematrix \cdot \cyclebasis \cyclelaplacian_{\cyclebasis}^{-1} \cyclebasis^{\Transpose} \cdot \d\edgematrix \right]_{i,j}
    = \frac{(-1)^{i+j} \dodgson^{e_i, e_j}_{\Graph}}{\dodgson_{\Graph}} \d a_{e_i} \wedge \d a_{e_j}.
  \end{align*}
  In particular, the square $\loopnumber \times \loopnumber$ submatrix $X[\overline{T},\overline{T}]$ consists of only the rows and columns corresponding to the edges in $\overline{T}$, ordered in increasing order.
  Denote the ordered edges of $\overline{T}$ by $\{f_1 < \cdots < f_{\loopnumber} \}$.
  Then expanding the Pfaffian of $X[\overline{T},\overline{T}]$ via \cref{eq:pfaff-defn} gives
  \begin{align}\label{eq:pff-tree-sub}
    &\Pf\left(X[\overline{T},\overline{T}] \right) \nonumber \\
    &= \frac{1}{2^{\frac{\loopnumber}{2}} \left(\frac{\loopnumber}{2}\right)!} \sum_{\sigma\in\mathfrak{S}_{\overline{T}}} \sgn \sigma \cdot  (-1)^{\sum_{i=1}^{\loopnumber} \sigma(f_i)} \cdot \frac{\dodgson_{\Graph}^{\sigma(f_1),\sigma(f_2)}}{\dodgson_{\Graph}} \cdots \frac{\dodgson_{\Graph}^{\sigma(f_{\loopnumber-1}),\sigma(f_{\loopnumber})}}{\dodgson_{\Graph}} \d a_{\sigma(f_1)} \wedge \cdots \wedge \d a_{\sigma(f_\loopnumber)} \nonumber \\
    &= \frac{1}{2^{\frac{\loopnumber}{2}} \left(\frac{\loopnumber}{2}\right)! \cdot \dodgson_{\Graph}^{\frac{\loopnumber}{2}}} (-1)^{\sum_{e \not\in T} e} \left(\sum_{\sigma\in\mathfrak{S}_{\overline{T}}} \dodgson_{\Graph}^{\sigma(f_1),\sigma(f_2)} \cdots \dodgson_{\Graph}^{\sigma(f_{\loopnumber-1}),\sigma(f_{\loopnumber})} \right) \bigwedge_{e \not\in T} \d a_{e},
  \end{align}
  where the sum runs over all permutations of the edges $\{f_1, \cdots f_\loopnumber\}$ and the second equality follows from noticing that $\sgn\sigma$ is exactly the sign needed to rearrange $\d a_{\sigma(f_1)} \wedge \cdots \wedge \d a_{\sigma(f_\loopnumber)}$ into increasing order.

  As this holds for every spanning tree $T$, plugging \cref{eq:pff-tree-sub} and \cref{eq:pff-minor-expn}
  into \cref{eq:pff-recall} gives the result.
\end{proof}

Note that using~\cref{lem:dual-inverse} and pulling in some prefactors, $\phi_{\Graph}$ can also be written as
\begin{align*}
\phi_{\Graph} = \frac{1}{(-2\pi)^{\loopnumber / 2} \sqrt{\firstsymanzik_{\Graph}}} \sum_{\substack{T \\ \text{spanning tree}}} \det\left(\cycleincidencematrix[\overline{T}]\right) \haf\left(\duallaplacian_{\fundamentalcbasis}^{-1}\right) \bigwedge_{e \notin T} \d a_e ,
\end{align*}
where $\fundamentalcbasis$ is the fundamental cycle basis (\cref{def:fundamentalcyclebasis}) corresponding to the spanning tree $T$ and $\haf(M)$ denotes the \emph{hafnian} of an even $2n \times 2n$ dimensional symmetric matrix $M$ defined as
\begin{align*}
 \haf(M) \defas \frac{1}{2^n n!} \sum_{\sigma \in \mathfrak{S}_n} M_{\sigma(1), \sigma(2)} \cdots M_{\sigma(2n-1), \sigma(2n)}.
\end{align*}
Thus, the hafnian is for the Pfaffian (\cref{eq:pfaff-defn}) what the permanent is for the determinant.

\subsection{Proof of the main result}

With \cref{lem:phi_dodgson} and the constructions that lead to it, we have established all the ingredients needed for the proof of our main result, \cref{thm:main_thm}. Consequences of these results have been discussed already in \cref{sec:consequences}.

\begin{proof}
	By \cref{lem:dbarT}, the form $\alpha_\Graph$ is a sum of terms, each of which is indexed by one of the spanning trees $T\subseteq E_\Graph$, where the term corresponding to $T$ is a certain sum of Dodgson polynomials.
	On the other hand, by \cref{lem:phi_dodgson}, the form $\phi_\Graph$ is given by a very similar sum, up to two differences. Firstly, the overall scaling factor is different. Secondly, in place of $\det \left( \incidencematrix[T] \right) $, we have $(-1)^{\sum_{e\notin T} e }\det \left( \cycleincidencematrix[\overline T] \right)$. By \cref{lem:signs}, these two quantities coincide up to the overall sign $\det\concatm{\cycleincidencematrix}{\pathmatrix}$ which is independent of $T$.
\end{proof}

\printbibliography

\end{document}